\documentclass{article}

\usepackage{arxiv}

\usepackage[utf8]{inputenc} 
\usepackage[T1]{fontenc}    
\usepackage{hyperref}       
\usepackage{url}            
\usepackage{booktabs}       
\usepackage{amsfonts}      
\usepackage{amsthm}  
\usepackage{nicefrac}       
\usepackage{microtype}      
\usepackage{lipsum}
\usepackage{graphicx}
\graphicspath{ {./images/} }

\usepackage{float}
\usepackage{amsmath}
\usepackage{graphicx}
\usepackage{natbib}
\usepackage{caption}
\usepackage{subcaption}
\usepackage[table]{xcolor} 

\newtheorem{theorem}{Theorem}[section]
\newtheorem{proposition}[theorem]{Proposition}

\title{Extensions in semiparametric geostatistical models}

\author{
 Maíra Soalheiro \\
  Departamento de Estatística\\
  Universidade Federal de Minas Gerais \\
  \texttt{msoalheiro@ufmg.br } \\
   \And
 Marcos Oliveira Prates \\
  Departamento de Estatística\\
  Universidade Federal de Minas Gerais \\
  \texttt{marcosop@ufmg.br} \\
  \And
 Victor Hugo Lachos \\
  Department of Statistics\\
  University of Connecticut \\
  \texttt{hlachos@uconn.edu } \\
   \And
 Fábio Nogueira Demarqui \\
  Departamento de Estatística\\
  Universidade Federal de Minas Gerais \\
  \texttt{fndemarqui@est.ufmg.br} \\
}

\begin{document}
\maketitle
\begin{abstract}
In spatial statistics, the incorrect selection of an appropriate covariance function may lead to inference errors and confidence underestimation. Motivated by such restrictions, we introduce and evaluate a flexible semiparametric approach for estimating spatial covariance functions based on Bernstein polynomials. The proposed formulation is general and applicable to classes of models that incorporate latent spatial effects in georeferenced data, such as Spatial Generalized Linear Mixed Models. Empirical validation was conducted via Monte Carlo simulations, and model fitting was performed using Bayesian inference via Markov chain Monte Carlo. Simulated scenarios demonstrated the model's ability to recover structural covariance configurations with low bias and high parameter precision. The practical applicability of the methodology was tested using real abundance data for American Robin (\textit{Turdus migratorius}) from the North American Breeding Bird Survey. The proposed model, featuring a Negative Binomial structure, yielded satisfactory results, efficiently capturing the overdispersion inherent in the count data. The estimated range parameter of 381.48 km revealed that the species' spatial dependence operates at a regional scale, suggesting that unobserved ecological processes act homogeneously within this radius of environmental influence. Additionally, predictive validation using an independent sample ($n_{\text{pred}} = 34$) demonstrated the model's strong generalization capability via Bayesian Kriging, producing point projections that closely matched observed values and well-calibrated prediction intervals. It is concluded that the proposed approach represents a robust methodological advancement, establishing itself as a flexible and efficient tool.
\end{abstract}

\keywords{Bernstein polynomials, Spatial covariance functions, Spatial Generalized Linear Mixed Models, Bayesian Kriging, Semiparametric spatial models, Geostatistics.}

\section{Introduction}

Spatial statistics plays a fundamental role in analyzing complex phenomena across various fields, such as epidemiology, agronomy, health, and ecology. It is characterized by the exploration of datasets exhibiting spatial or spatio-temporal dependence, a structure that, if overlooked, compromises the quality of inference. Among the classes of spatial data, this work focuses on geostatistical (or point-level) data, which aim to quantify the relationship between variables of interest based on their geographic locations \citep{banerjee2015book}. Traditionally, modeling these stochastic processes relies on the classical geostatistical formulation using isotropic parametric covariance functions, such as the Matérn and power exponential \citep{cressie1993, diggle2007}. Although analytically appealing, these classical families may prove insufficient to capture the stochastic complexity of certain natural phenomena, potentially leading to systematic inference errors and an underestimation of predictive uncertainty in kriging maps.

To overcome these limitations, non-parametric and semi-parametric approaches based on kernel smoothing and spectral representations have gained prominence \citep{hall1994nonparametric, genton2002nonparametric}. From a Bayesian perspective, Non-Parametric Bayesian Geostatistics (NPBG) utilizing Spatial Dirichlet Processes \citep{gelfand2005spatial}, has emerged as a robust solution for capturing non-stationarity and asymmetries. However, these techniques often face high computational complexity, difficulties in enforcing fundamental theoretical constraints (such as the positive definiteness of the covariance matrix), and slow convergence of Markov chain Monte Carlo (MCMC) algorithms when applied to large datasets. In this context, approaches based on polynomial approximations emerge as a mathematically simple and computationally feasible alternative. In particular, Bernstein polynomials (BPs) have gained traction in the statistical literature due to the ease of directly managing shape constraints via simple linear constraints on their coefficients \citep{book:5337, manchuk2008experimental}.

Recently, building on this line of research, \citet{wang2023nonparametric} introduced, in the spectral domain, a nonparametric, isotropic covariance function based on Bernstein polynomials for Gaussian processes within the frequentist paradigm. However, the author did not fully explore intrinsic properties of BPs that are particularly advantageous for Bayesian modeling, as imposing monotonicity and convexity leads to significantly more stable and efficient MCMC sampling compared to traditional techniques based on mixtures or spectral densities, allowing for straightforward implementation in probabilistic programming platforms such as \texttt{rstan} \citep{rstan2026}. Furthermore, the current literature lacks flexible hierarchical extensions that transition from the strictly linear Gaussian case to more complex modeling structures.

In this article, we propose a flexible and computationally efficient Bayesian semiparametric approach for estimating spatial covariance functions based on Bernstein polynomials on the original observational domain. Unlike the approach of \citet{wang2023nonparametric}, which is in the spectral domain and limited to the classical linear Gaussian model, the proposed formulation is general and applicable to any class of models incorporating latent spatial effects in point-level data, such as spatial generalized models, spatial survival models, or spatial quantile regression. To demonstrate the method's practical applicability and flexibility in non-Gaussian scenarios, the methodology is presented in the context of Generalized Linear Mixed Models (GLMMs) and applied to the analysis of real data from the Breeding Bird Survey (BBS), considering spatial modeling under Poisson, Negative Binomial, and Bell distributions. We demonstrate that models fitted using the proposed covariance function recover the parameters of interest and appear to be the best option when the true data covariance is unknown. Furthermore, our alternative is shown to be robust to model misspecification, providing reliable prediction estimates and intervals and thus demonstrating good generalization capability.

The remainder of this article is organized as follows. Section \ref{sec:preliminares} reviews the fundamental concepts of spatial Gaussian processes, Bernstein polynomials, and spatial hierarchical models. Section \ref{sec:metodologia} presents the proposed non-parametric spatial covariance model. Section \ref{sec:sim} presents a simulation study to evaluate the performance of our approach. In Section \ref{sec:aplicacao}, the proposed framework is applied to real spatial data. Finally, Section \ref{sec:conclusao} provides concluding remarks and directions for future research.


\section{Preliminaries}\label{sec:preliminares}


\subsection{Spatial Gaussian Processes}

Let $\{Y(\mathbf{s}): \mathbf{s} \in \boldsymbol{\mathcal{D}}\}$ be a spatial stochastic process defined on a continuous domain $\boldsymbol{\mathcal{D}} \subset \mathbb{R}^r$. The process is termed a Gaussian Process (GP) if, for any set of locations $\{\mathbf{s}_1, \dots, \mathbf{s}_n\}$, the observed vector $\mathbf{Y} = (Y(\mathbf{s}_1), \dots, Y(\mathbf{s}_n))^{\top}$ follows a multivariate normal distribution $\mathbf{Y} \sim \mathcal{N}_n(\boldsymbol{\mu}, \mathbf{\Sigma})$, with mean vector $\boldsymbol{\mu}$ and spatial covariance matrix $\mathbf{\Sigma}$ \citep{banerjee2015book}.

Assuming second-order stationarity and isotropy, the mean is constant ($E(Y(\mathbf{s})) = \mu$) and the covariance between two points depends exclusively on the Euclidean distance $\Vert{}\mathbf{d}\Vert{} = \Vert{}\mathbf{s}_i - \mathbf{s}_j\Vert{}$, expressed by the kernel function $C(\Vert{}\mathbf{d}\Vert{}) = Cov(Y(\mathbf{s}_i), Y(\mathbf{s}_j))$.

The structure of the covariance function $C(\Vert{}\mathbf{d}\Vert{})$ governs the stochastic behavior and differentiability of the process curves; thus, incorrect kernel specifications can lead to inference errors and distortions in smoothness. Usually, $C(\Vert{}\mathbf{d}\Vert{})$ is parameterized by: $\sigma^2 > 0$: spatial variance (partial sill); $\tau^2 \ge 0$: nugget effect, representing measurement error or micro-scale variability; and $\phi > 0$: decay parameter governing the spatial range. For $C(\Vert{}\mathbf{d}\Vert{})$ to define a valid GP, it must be a positive-definite function, ensuring that:
$$\sum_{i=1}^n \sum_{j=1}^n a_i a_j C(\mathbf{s}_i - \mathbf{s}_j) > 0$$
for any non-zero constant vector $\mathbf{a}$. This guarantee of strict positivity ensures that the spatial covariance matrix $\mathbf{\Sigma}$ is non-singular, invertible, and well-conditioned, an indispensable requirement for computing the Gaussian likelihood and for sampling stability via MCMC methods.

\subsection{Bernstein Polynomials}

Bernstein polynomials provide a flexible basis for the smooth, non-parametric approximation of continuous functions on a bounded interval, grounded in the Weierstrass Approximation Theorem. 
For a function $f(x)$ defined on $x \in [0, 1]$, the Bernstein polynomial of order $m$ is defined as a linear combination of $m+1$ basis functions $b_{k,m}(x)$: $$B_m(x) = \sum_{k=0}^{m} f\left(\frac{k}{m}\right) b_{k,m}(x), \quad \text{where} \quad b_{k,m}(x) = \binom{m}{k} x^k (1-x)^{m-k}$$ as $m \to \infty$, $B_m(x)$ converges uniformly to $f(x)$ \citep{book:5337}. The basis functions $b_{k,m}(x)$ possess analytical properties that are critical for stochastic modeling: non-negativity and Partition of Unity: $b_{k,m}(x) \ge 0$ and $\sum_{k=0}^{m} b_{k,m}(x) = 1$, $\forall x \in [0,1]$; symmetry: $b_{k,m}(x) = b_{m-k,m}(1-x)$; derivative structure: the $r$-th derivative of $B_m(x)$ is expressed in terms of basis functions of order $m-r$: $$B^{(r)}_m(x) = \frac{m!}{(m-r)!} \sum_{k=0}^{m-r} \Delta^r f\left(\frac{k}{m}\right) b_{k,m-r}(x)$$ where $\Delta$ denotes the forward difference operator; positive Linear Operator: preserves the linearity and the monotonicity/positivity of the original function. Compared to other bases (such as B-splines), the PB offers significant advantages for Bayesian inference  \citep{ghosal2017fundamentals}: shape constraints (positivity, monotonicity, and convexity) translate into straightforward linear constraints on the coefficients. Furthermore, the natural link to the Beta density simplifies the specification of prior distributions in MCMC algorithms and ensures the analytical validity of the approximating function. 

According to \cite{osman2012nonparametric}, given a continuous function $F(t)$ on $t \in (0, \tau]$, the rescaled $m$-th order Bernstein polynomial (BP) approximation is given by $B_m(t; F) = \sum_{k=0}^m b_{k,m} B_{k,m}(t)$, where $B_{k,m}(t) = \binom{m}{k} (t/\tau)^k (1 - t/\tau)^{m-k}$. The first derivative $B'_m(t; F) = b_m(t; F)$ expresses a direct relationship with the Beta distribution density $f_B$: $$b_m(t; F) = \sum_{k=1}^{m} \left[ F\left(\frac{k}{m}\tau\right) - F\left(\frac{k-1}{m}\tau\right) \right] \frac{f_B(t/\tau \mid k, m-k+1)}{\tau}$$ To model a function $f(h; \boldsymbol{\gamma})$ with $h \ge 0$, the representation $f(h; \boldsymbol{\gamma}) = \sum_{k=1}^m \gamma_k g_{k,m}(h)$ is adopted, where $\boldsymbol{\gamma} = (\gamma_1, \dots, \gamma_m)^\top$ is the vector of coefficients with $\gamma_k \ge 0$, and $g_{k,m}(h)$ represents the Beta density evaluated at $h/\tau$. Integrating $f(h; \boldsymbol{\gamma})$ yields the approximate cumulative distribution function: $$F(h; \boldsymbol{\gamma}) = \sum_{k=1}^{m} \gamma_k G_{k,m}(h) = \boldsymbol{\gamma}^\top \mathbf{G}_m(h)$$ where $G_{k,m}(h) = F_B(h/\tau \mid k, m-k+1)$ is the c.d.f. and $\mathbf{G}_m(h)$ is a ...c.d.f. and $\gamma_k \ge 0$, the resulting function $F(h; \boldsymbol{\gamma})$ is monotonically increasing by construction.

\subsection{Spatial Hierarchical Model}\label{sec:gshm}

Consider a spatial hierarchical model, and let $\mathbf{s}_i \in \mathcal{D} \subset \mathbb{R}^2$ be the exact geographic coordinates of the $i$-th sampling location, for $i = 1, \ldots, n$. The first level of the hierarchical structure represents the conditional data model and is defined as:
\begin{equation*}
\label{eq:spatial_data_model_gen}
Y(\mathbf{s}_i) \mid \theta(\mathbf{s}_i), \phi \sim f\left(Y(\mathbf{s}_i) \mid \mu(\mathbf{s}_i), \phi\right),
\end{equation*}
where $Y(\mathbf{s}_i)$ denotes the response variable observed at location $\mathbf{s}_i$, and $\phi$ is the dispersion or shape parameter associated with the chosen distribution family $f(\cdot)$. The conditional mean $\mu(\mathbf{s}_i) = E(Y(\mathbf{s}_i))$ is related to the global intercept, the covariates, and the latent spatial effect via a link function $g(\cdot)$:
\begin{equation*}
\label{eq:spatial_linear_predictor_gen}
g(\mu(\mathbf{s}_i)) = \beta_{0} + \mathbf{x}(\mathbf{s}_i)^T \boldsymbol{\beta} + \theta(\mathbf{s}_i),
\end{equation*}
where $\beta_{0}$ represents the global intercept, $\mathbf{x}(\mathbf{s}_i)$ is the vector of covariates observed at coordinate $\mathbf{s}_i$ associated with the coefficient vector $\boldsymbol{\beta}$, and $\theta(\mathbf{s}_i)$ is the latent effect corresponding to spatial dependence.

The second level is given by the spatial component $\theta(\mathbf{s})$; this is modeled as a GP with zero mean and covariance matrix $\boldsymbol{\Sigma}$. This implies that the joint vector of latent spatial effects evaluated at the $n$ sampled locations, $\boldsymbol{\theta} = (\theta(\mathbf{s}_1), \cdots, \theta(\mathbf{s}_n))^T$, follows a multivariate normal distribution:
\begin{equation*}
\label{eq:spatial_process_model_gen}
\boldsymbol{\theta} \mid \boldsymbol{\psi} \sim \mathcal{N}_m\left(\mathbf{0}, \boldsymbol{\Sigma}(\boldsymbol{\psi})\right),
\end{equation*}
where $\boldsymbol{\Sigma}$ is a valid covariance matrix.

\section{Methodology}
\label{sec:metodologia}

Covariance functions form the fundamental basis of geostatistical modeling and Gaussian processes, characterizing the dependence between observations indexed by geographic coordinates \citep{cressie1993}. Although widely used, classical parametric covariance families, such as the exponential, spherical, or Matérn functions, impose rigid constraints on the rate of correlation decay due to their fixed mathematical specifications; consequently, they may fail to adapt to the complex, non-linear behaviors common in real-world data \citep{stein1999interpolation}. To overcome these limitations, non-parametric approaches have gained prominence in recent literature \citep{hall1994nonparametric, wang2023nonparametric}.
Next, we introduce a Bayesian semi-parametric approach based on BPs, which are recognized for their excellent properties regarding smooth function approximation and the handling of shape constraints \citep{book:5337, ghosal2001posterior}. 


\subsection{The Geo-Bernstein Polynomial Covariance Function}

The main contribution of this article is the specification of an isotropic, non-parametric covariance function termed the Geo-Bernstein Polynomial (GBP). Let $\tau$ be the maximum observed distance within the spatial domain under study. We define the proposed covariance function $C(d; \boldsymbol{\gamma}, \sigma^2)$ as: $$C(d;\boldsymbol{\gamma}, \sigma^2)= \left\{\begin{matrix}
\sigma^{2} \exp(- \boldsymbol{\gamma}^{\top} \mathbf{G}_{m}(d)),  & \text{if } 0 \leq d < \tau\\
\sigma^2 \exp \left( - \left(\sum_{k=1}^m  \gamma_k + \frac{m\gamma_m}{\tau} (d-\tau)\right)\right), & \text{if } d \geq \tau,
\end{matrix}\right.$$ where $\sigma^2$ is the conditional spatial variance (partial sill), and $\boldsymbol{\gamma} = (\gamma_{1}, \cdots, \gamma_{m})^{\top}$ is a vector of unknown polynomial coefficients subject to a non-negativity constraint ($\gamma_{k} \geq 0$ for $k = 1, \cdots, m$). The polynomial basis vector is denoted by $\mathbf{G}_{m}(d) = [G_{1}(d), \cdots, G_{m}(d)]^{\top}$, where each element $G_{k}(d) = F_{B}(d/\tau \mid k, m-k+1)$ corresponds to the cumulative distribution function (CDF) of a $Beta(k, m-k+1)$ distribution evaluated at $d/\tau$. The property of decreasing monotonicity of $C(d; \boldsymbol{\gamma}, \sigma^2)$ is guaranteed by the constraints $\boldsymbol{\gamma} \geq 0$ and the non-decreasing nature of $\mathbf{G}_{m}(d)$ for $d > 0$. This behavior preserves Tobler's First Law of Geography \citep{tobler1970computer}, analytically ensuring that:$$\lim_{d \to 0} C(d; \boldsymbol{\gamma}, \sigma^2) = \sigma^2 \quad \text{and} \quad \lim_{d \to +\infty} C(d; \boldsymbol{\gamma}, \sigma^2) = 0.$$Below, we present the fundamental analytical properties that validate the GBP structure within classical geostatistical theory.


\textbf{Property 1 (Behavior at the Origin)}. At zero distance ($d=0$), the GBP covariance function reduces exactly to the process variance, $C(0; \boldsymbol{\gamma}, \sigma^2) = \sigma^2$.

Proof: When $d=0$, we have $d/\tau = 0$. Since the CDF of any Beta distribution with strictly positive parameters evaluated at the origin is zero, it follows that $\mathbf{G}_{m}(0) = \mathbf{0}$. Substituting this result into the model definition yields $C(0; \boldsymbol{\gamma}, \sigma^2) = \sigma^2 \exp(-\boldsymbol{\gamma}^{\top} \cdot \mathbf{0}) = \sigma^2$, ensuring the process's consistency without the need to explicitly include a nugget effect in the latent spatial term. 

\begin{proposition}
\label{prop1}
(Ergodicity). \textit{ The spatial process associated with the GBP covariance structure is ergodic. That is, assuming $\gamma_k \geq 0$ for all $k$ and that at least one coefficient is strictly positive ($\gamma_k > 0$), the correlation function $\rho(d) = C(d) / \sigma^2$ satisfies: $$\lim_{d \to \infty} \rho(d; \boldsymbol{\gamma}) = 0$$}
\end{proposition}

\begin{proof}
Given that the threshold $\tau$ is finite ($\tau < \infty$), for any sufficiently large distance $d$ such that $d \ge \tau$, the correlation function $\rho(d; \gamma)$ is governed by the second branch of the covariance model:
$$ \rho(d; \gamma) = \exp \left( - \left( \sum_{k=1}^{m} \gamma_k + \frac{m \gamma_m}{\tau}(d - \tau) \right) \right). $$
To demonstrate that the spatial process is ergodic, we evaluate the asymptotic behavior of the correlation function as $d \to \infty$. Let $g(d)$ denote the argument of the exponential function:
$$ g(d) = \sum_{k=1}^{m} \gamma_k + \frac{m \gamma_m}{\tau}(d - \tau). $$
By hypothesis, $\gamma_k \ge 0$ for all $k$, and $\tau > 0$. Assuming the terminal coefficient $\gamma_m$ is strictly positive ($\gamma_m > 0$) to ensure spatial decay, the linear term in $d$ diverges to positive infinity. Consequently, we have:
$$ \lim_{d \to \infty} g(d) = \infty. $$
Given the continuity and asymptotic properties of the exponential function, it follows that:
$$ \lim_{d \to \infty} \rho(d; \gamma) = \lim_{d \to \infty} \exp(-g(d)) = 0. $$
Therefore, the correlation strictly dissipates as the distance approaches infinity. This asymptotic decay satisfies the sufficient condition for the ergodicity of the associated spatial process.
\end{proof}


\begin{proposition}(Equivalence to the Exponential Model). \textit{The classical parametric exponential covariance model is a special case of the proposed GBP structure when the polynomial has degree $m = 1$.}
\end{proposition}

\begin{proof} Setting the polynomial degree to $m = 1$, the coefficient vector reduces to a single scalar $\boldsymbol{\gamma} = \gamma_1$. The basis function $\mathbf{G}_{1}(d)$ becomes the CDF of a $Beta(1, 1)$ distribution. Since the $Beta(1,1)$ distribution is a standard Uniform distribution $U(0,1)$, its CDF is the identity function itself; that is, $F_{B}(d/\tau \mid 1, 1) = d/\tau$. 
Under these boundary conditions, the first branch of the covariance function reduces to:$$C(d; \gamma_1) = \sigma^2 \exp\left( - \gamma_1 \frac{d}{\tau} \right) = \sigma^2 \exp\left( - \phi d \right),$$where $\phi = \gamma_1/\tau$ acts analogously to the range parameter. Thus, the classical parametric exponential model emerges as a special case of the proposed formulation.
\end{proof}

By increasing the polynomial degree ($m > 1$), the GBP model gains the mathematical flexibility to dynamically learn complex patterns, asymmetries, and nonlinear correlation decay directly from the data, overcoming the rigidity limitations of classical functions while preserving the rigor and essential asymptotic properties of geostatistical theory.

\section{Simulation studies}
\label{sec:sim}
In this section, we evaluate the goodness-of-fit of the proposed semiparametric model for cases where data were generated from an isotropic covariance function. The primary objective is to analyze the ability of the GBP model to approximate the theoretical correlation curve and recover the true process parameters, comparing its performance with classical approaches.


To test the flexibility of the GBP model regarding different degrees of smoothness and spatial decay behaviors, data were generated from four classical covariance structures: Exponential, Matérn $3/2$, Power Exponential with shape parameter $p=0.5$ (generating a rougher surface), and Power Exponential with shape parameter $p=1.8$ (generating a smoother surface). For each scenario, we fitted three distinct covariance functions: our proposed non-parametric model (GBP), the exponential covariance, and the Matérn $3/2$ covariance. 

The data were generated using a spatial linear regression model, which is a special case of the hierarchical structure in Section~\ref{sec:gshm}, to preserve the direct interpretability of the parameters. The response variable vector $\mathbf{Y} \in \mathbb{R}^n$ is defined by:
$$
\mathbf{Y} = \beta_{0} + \mathbf{X}\boldsymbol{\beta} + \boldsymbol{\theta},
$$
where $\beta_{0}$ represents the global intercept, $\mathbf{X}$ is the covariate matrix associated with the coefficient vector $\boldsymbol{\beta}$, and $\boldsymbol{\theta} = (\theta_1, \cdots, \theta_n)^{\top}$ denotes the spatial dependence random effect, modeled via a Gaussian Process with zero mean and covariance matrix $\boldsymbol{\Sigma}$, such that $\boldsymbol{\theta} \sim \mathcal{GP}(\mathbf{0}, \boldsymbol{\Sigma})$. The true values set for the simulation were: $\sigma^2 = 1$ (partial sill), $\beta_0 = 1$, and $\beta_1 = -0.75$.


The analyses were performed in the \texttt{R} statistical environment \citep{R} using the \texttt{rstan} interface. One hundred Monte Carlo (MC) replicates were conducted for three distinct sample sizes ($n = 100, 200, 500$). The sampling process followed a nested design: initially, a sample of 500 observations randomly distributed within a unit square $[0,1] \times [0,1]$ was generated. From this set, a subsample of 200 observations was drawn, and from that, a further subsample with $n=100$ was obtained. Following literature recommendations for non-parametric polynomial models \citep{osman2012nonparametric}, the degree of the Bernstein polynomial ($m$) was defined proportionally to the sample size, adopting the empirical rule $m \approx \sqrt{n}$. Therefore, $m=10$, $m=14$, and $m=22$ were set for sample sizes $n=100$, $n=200$, and $n=500$, respectively. The spatial scale parameter $\phi$ used in data generation was specified individually for each covariance function. The criterion adopted ensured that the theoretical spatial correlation decayed to exactly $0.05$ at a distance corresponding to $40\%$ of the maximum diagonal of the unit square (fixing the empirical practical range at $0.57$ units). The $\phi$ values are presented in Table~\ref{tab:valores_phi}.
\begin{table}[H]
\centering
\caption{Values of $\phi$ for covariance function.}
\label{tab:valores_phi}
\begin{tabular}{l*{2}{c}r}
\hline
Covariance function & $\phi$ \\
\hline
Exponential & $0.189 $ \\
Matérn 3/2 & $0.207$ \\
Power exp ($p=0.5$) &  $0.063 $ \\
Power exp ($p=1.8$) &  $0.308$ \\
\hline
\end{tabular}
\end{table}


Parameter estimation was conducted within the Bayesian paradigm using the MCMC algorithm. Specifically, the Hamiltonian Monte Carlo (HMC) algorithm \citep{neal2011hmc} was employed in combination with the No-U-Turn Sampler (NUTS) \citep{hoffman2014nuts}. To ensure numerical stability and algorithm convergence, regression coefficients were standardized, and weakly informative prior distributions were assigned, following current modeling recommendations \citep{rstan2026}:
\begin{eqnarray*}
\beta_0^* &\sim& \mathcal{N}(0, 10),  \\
\beta^* &\sim& \mathcal{N}(0,3), \\
\sigma &\sim& \text{Half-Student-t}(\nu=3, 0, 1), \\
\phi &\sim& \text{Uniform}(0, b_{\phi}) \quad \text{or} \quad \boldsymbol{\gamma} \sim \text{Log-Normal}(\mathbf{0}, 4).
\end{eqnarray*}
Here, $\beta_0^*$ and $\beta^*$ denote the priors applied to the standardized coefficients. The upper bound of the uniform distribution for the decay parameter ($\phi$), $b_{\phi}$, was set to the value required to ensure that spatial correlation drops to $0.05$ at a maximum distance of $75\%$ of the unit square's diagonal. The $\text{Half-Student-t}$ distribution has $3$ degrees of freedom, is centered at $0$, and has a scale parameter of $1$. For the proposed GBP structure, the polynomial coefficients were assigned a Log-Normal prior, ensuring the non-negativity constraint ($\gamma_k \geq 0$).


\subsection{Exponential Covariance Scenario}
\label{sec:exp:sim}

In this scenario, data were generated using an exponential covariance matrix, based on the parametric specifications and the range parameter ($\phi$) presented in Table~\ref{tab:valores_phi}. For each Monte Carlo replicate, the Exponential, Matérn, and GBP models were fitted.
\begin{figure}[htb]
    \centering
    \begin{subfigure}{0.3\textwidth}
        \centering
        \includegraphics[width=\linewidth]{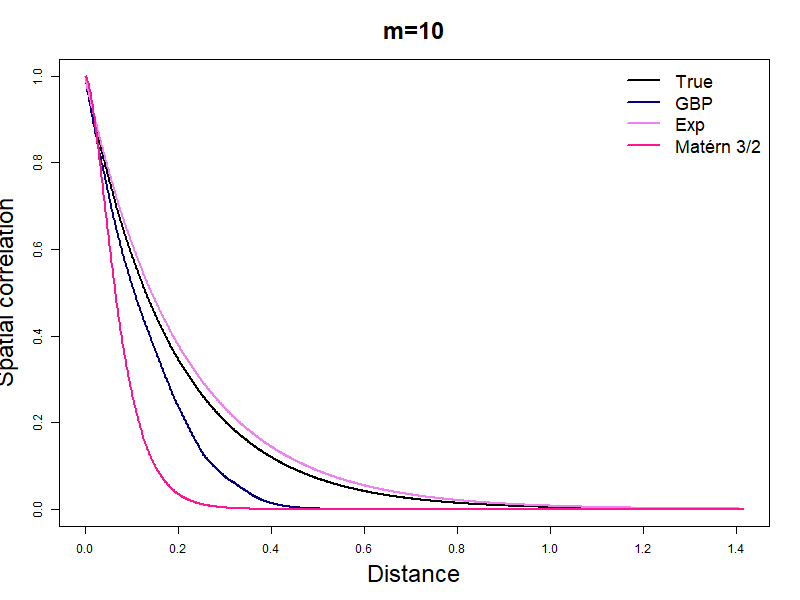}
        \caption{$n=100$}
    \end{subfigure}
    \hfill
    \begin{subfigure}{0.3\textwidth}
        \centering
        \includegraphics[width=\linewidth]{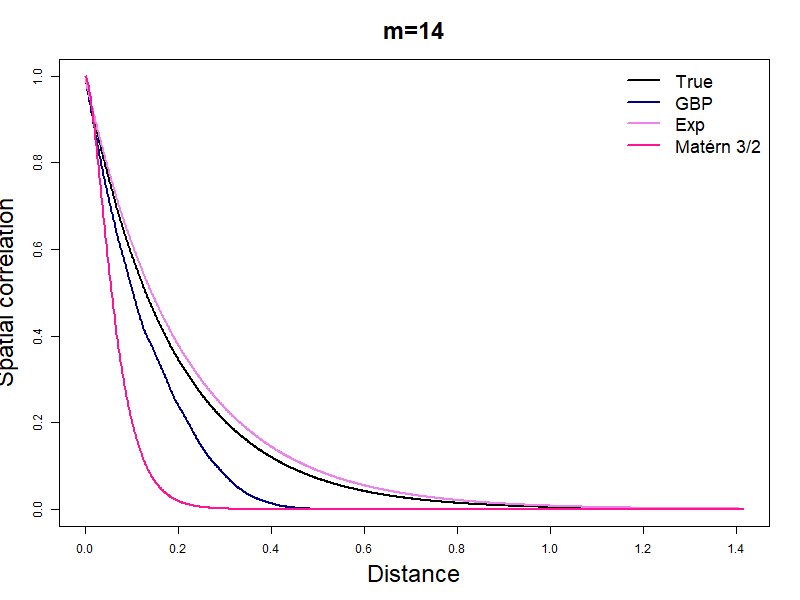}
        \caption{$n=200$}
    \end{subfigure}
    \hfill
    \begin{subfigure}{0.3\textwidth}
        \centering
        \includegraphics[width=\linewidth]{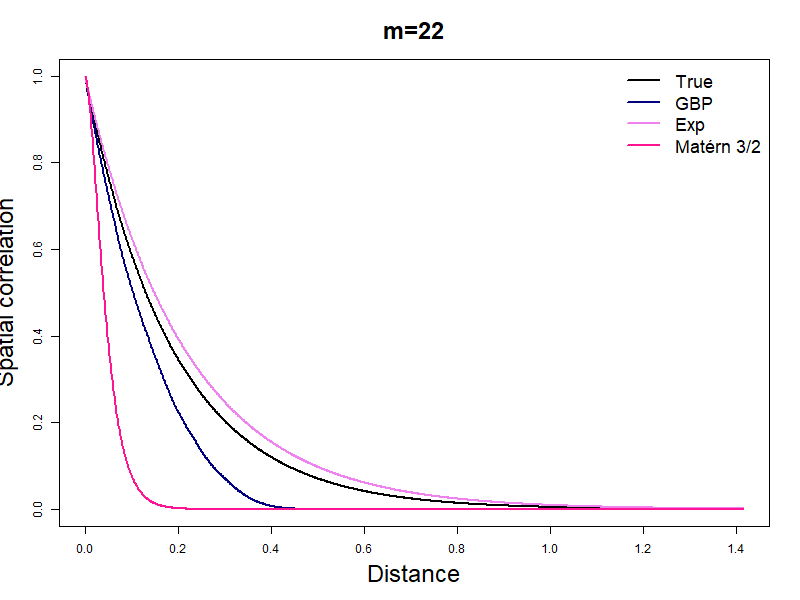}
        \caption{$n=500$}
    \end{subfigure}
    \caption{Spatial correlation curves: exponential covariance.}\label{fig:sim_exp}
\end{figure}

\begin{figure}[htb]
    \centering
    \begin{subfigure}{0.3\textwidth}
        \centering
        \includegraphics[width=\linewidth]{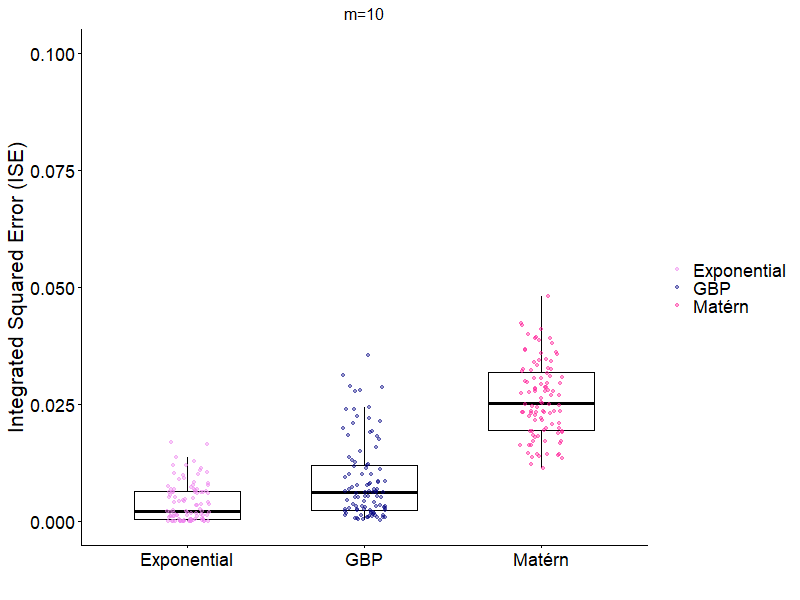}
        \caption{$n=100$}
    \end{subfigure}
    \hfill
    \begin{subfigure}{0.3\textwidth}
        \centering
        \includegraphics[width=\linewidth]{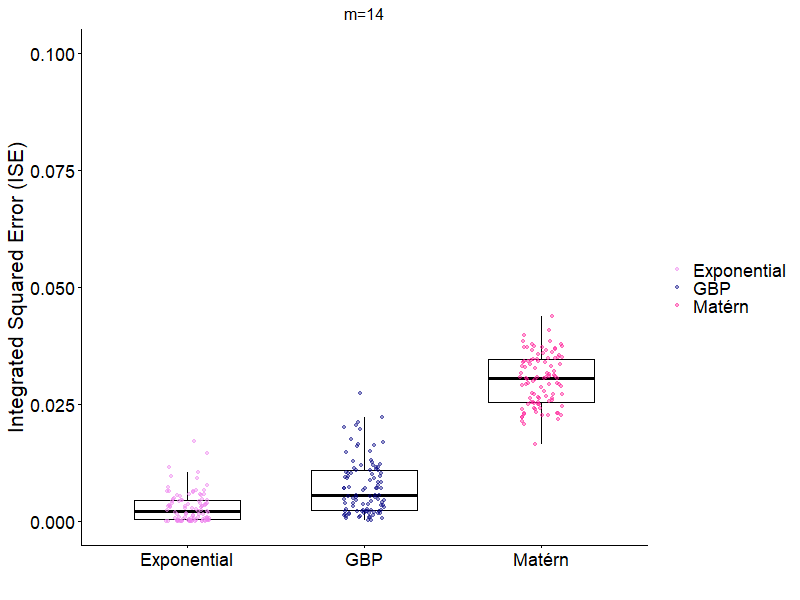}
        \caption{$n=200$}
    \end{subfigure}
    \hfill
    \begin{subfigure}{0.3\textwidth}
        \centering
        \includegraphics[width=\linewidth]{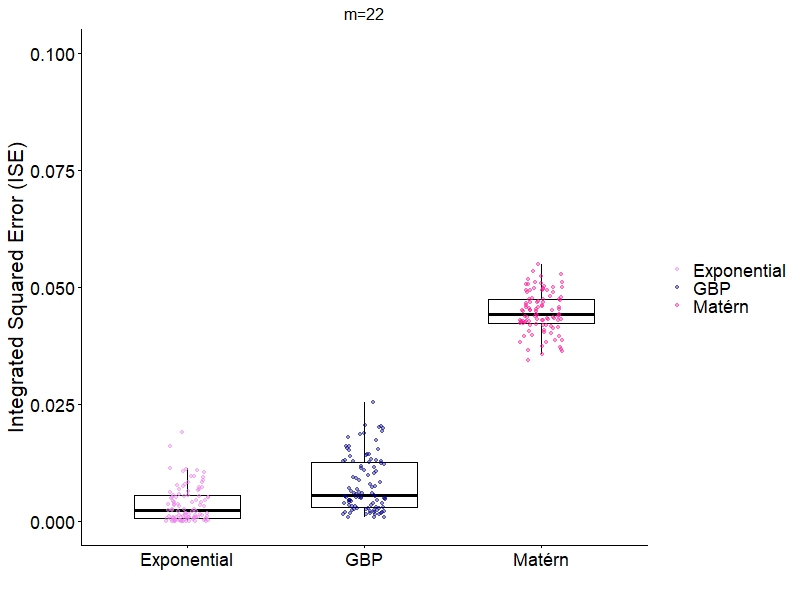}
        \caption{$n=500$}
    \end{subfigure}
    \caption{Boxplot of ISE: exponential covariance case.}\label{fig:dist_exp}
\end{figure}
Figure \ref{fig:sim_exp} shows the medians of the estimated spatial correlation curves, while Figure \ref{fig:dist_exp} presents the boxplot of the ISE. Across the three sample sizes analyzed, the curve estimated by the GBP model is the one that most closely approximates the true theoretical curve, surpassed only by the Exponential model, which represents the exact structure of the generating process. Additionally, the GBP model exhibited a less abrupt decay compared to that estimated by the Matérn function. This restrictive behavior of the Matérn model was theoretically expected. This class has greater flexibility regarding the process's differentiability and smoothness, therefore, it tends to underestimate the spatial range to accommodate the decay rate imposed by the exponential model, resulting in a relatively smaller practical range.

\begin{figure}[H]
    \centering
    \begin{subfigure}{\textwidth}
        \centering
        \includegraphics[height=5cm]{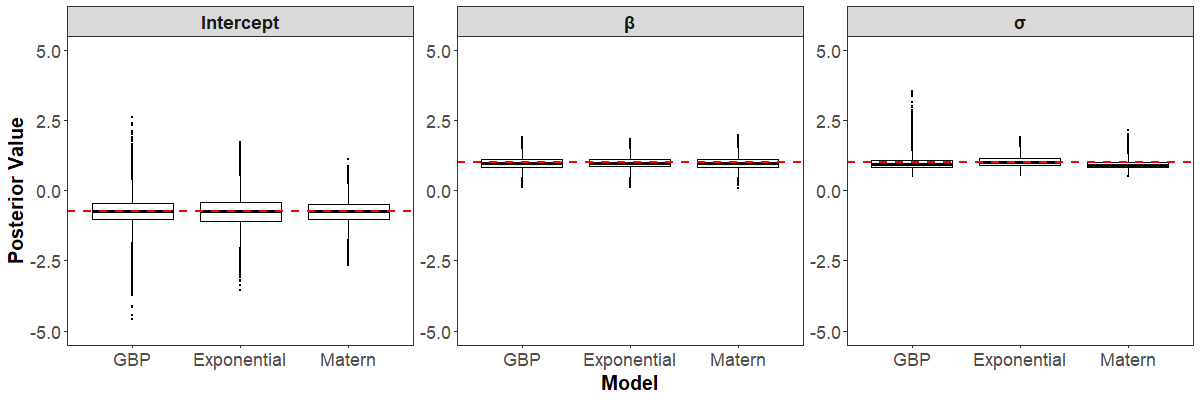}
        \caption{$n=100$.}
        \label{fig:param_exp100}
    \end{subfigure}
    \vspace{0.3cm} 
    
    \begin{subfigure}{\textwidth}
        \centering
        \includegraphics[height=5cm]{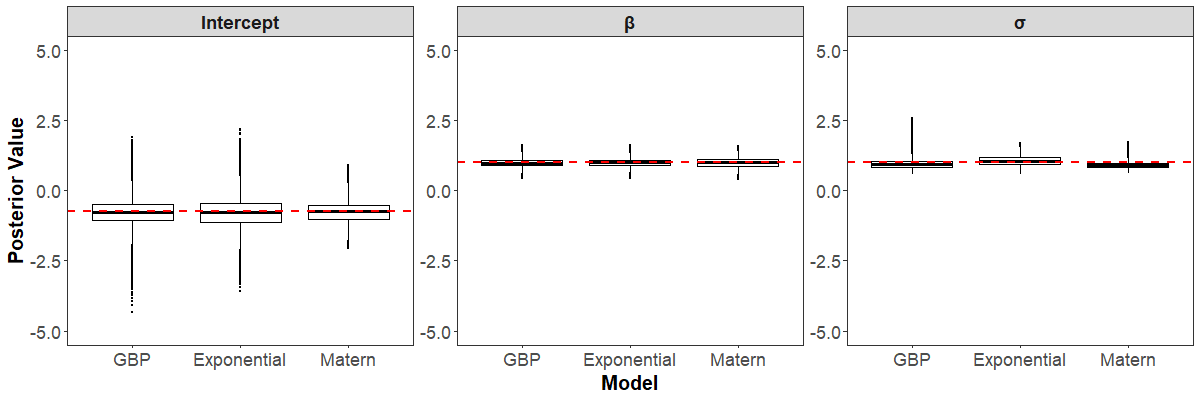}
        \caption{$n=200$.}
        \label{fig:param_exp200}
    \end{subfigure}
    \vspace{0.3cm}
    
    \begin{subfigure}{\textwidth}
        \centering
        \includegraphics[height=5cm]{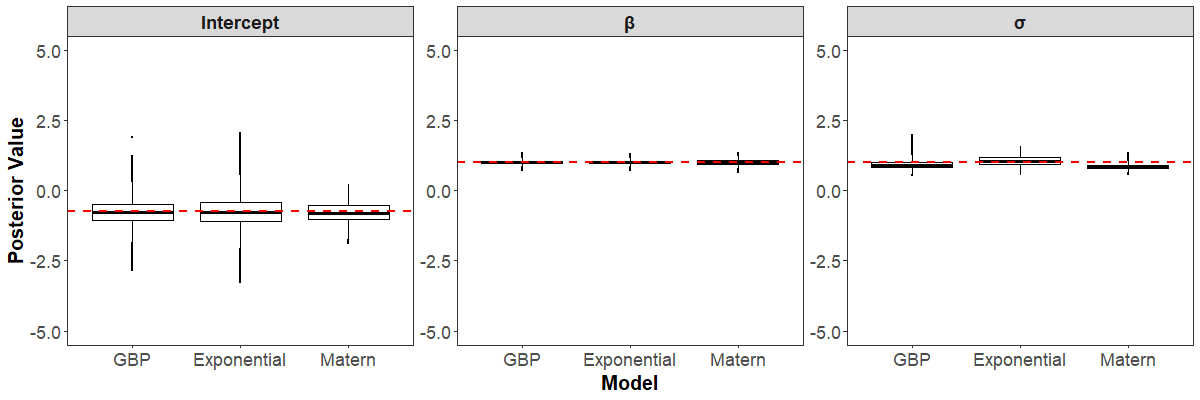}
        \caption{$n=500$.}
        \label{fig:param_exp500}
    \end{subfigure}
    
    \caption{Box plot of estimated parameters: exponential case.}
    \label{fig:param_exp_all}
\end{figure}

Figure~\ref{fig:param_exp_all} presents the accuracy and precision of the estimates obtained across the 100 replicates for sample sizes $n=100$, $n=200$, and $n=500$. An analysis of the boxplots in Figure \ref{fig:param_exp100} indicates that the three models yielded similar results in estimating the regression coefficients $\beta_0$ (intercept) and $\beta_1$ ($\beta$), with the posterior medians converging toward the true value. Regarding the parameter $\sigma$, the models approached the true value, showing little variation across iterations. Figure \ref{fig:param_exp200} shows that, even with the change in sample size, the models produced results similar to those of the previous sample for both the regression coefficients and the parameter $\sigma$. The relative quality of the recovery remained consistent: the exponential model yielded the posterior median closest to the true value, followed by the GBP and the Matérn function.

In general, these results remained consistent as the sample size increased (Figure \ref{fig:param_exp500}). In this case, it can be concluded that the posterior median estimates for the parameters are close to or equivalent to the model's true values, despite some dispersion-related distortions. Furthermore, the quality of the recovery based on posterior median fits reflects the performance of the estimated spatial correlation curves, with the GBP emerging as the best alternative to the model fitted using the data's true covariance.

\subsection{Scenario with Matérn 3/2 Covariance}

In the second scenario, the previous methodological specifications were maintained, but the data were generated using a Matérn 3/2 covariance matrix.  Figure ~\ref{fig:sim_matern} shows the medians of the estimated spatial correlation curves, and Figure \ref{fig:dist_matern} presents the boxplot of the ISE.

\begin{figure}[H]
    \centering
    \begin{subfigure}{0.3\textwidth}
        \centering
        \includegraphics[width=\linewidth]{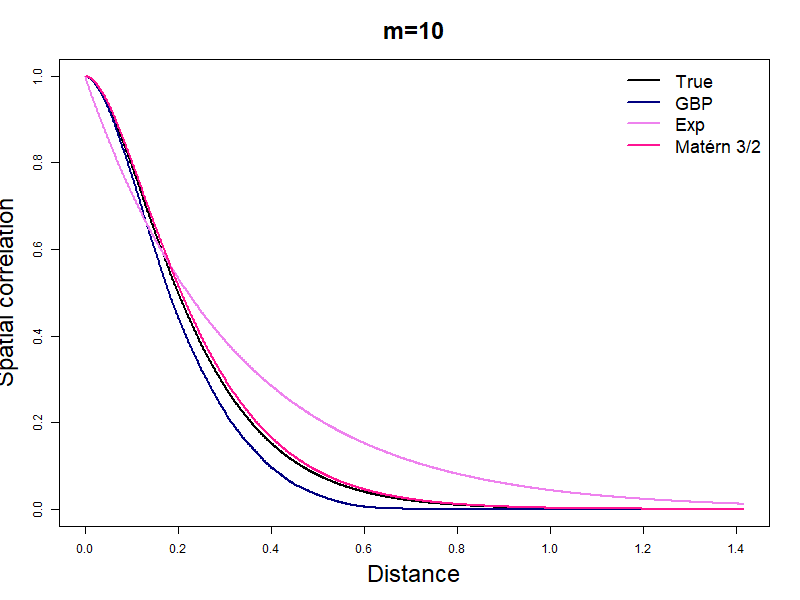}
        \caption{$n=100$}
    \end{subfigure}
    \hfill
    \begin{subfigure}{0.3\textwidth}
        \centering
        \includegraphics[width=\linewidth]{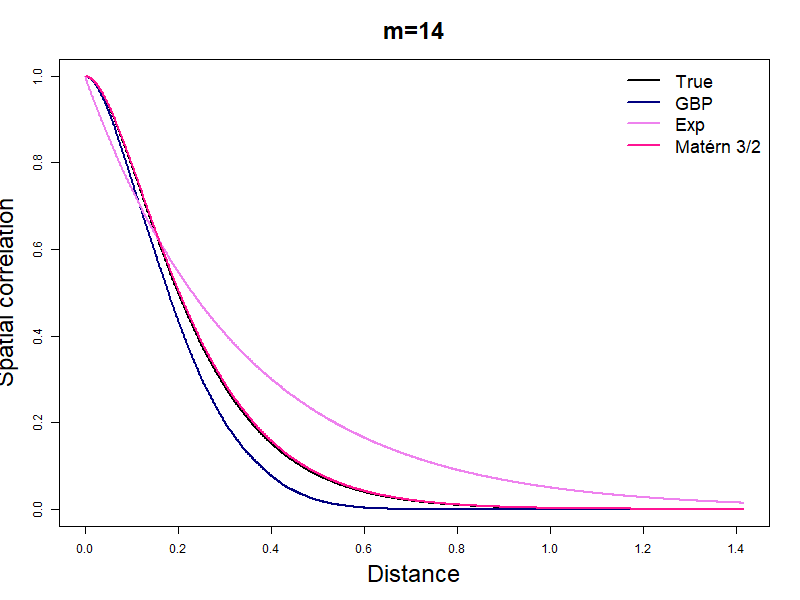}
        \caption{$n=200$}
    \end{subfigure}
    \hfill
    \begin{subfigure}{0.3\textwidth}
        \centering
        \includegraphics[width=\linewidth]{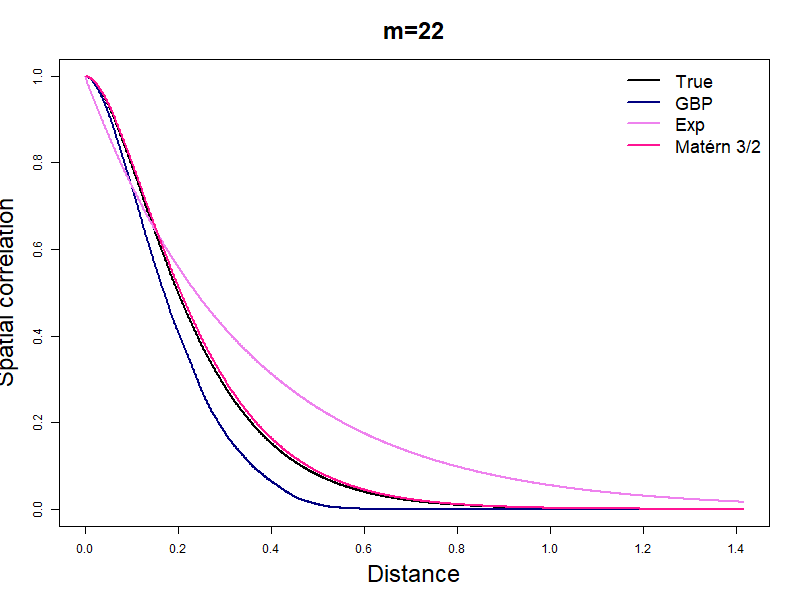}
        \caption{$n=500$}
    \end{subfigure}
    \caption{Spatial correlation curves: Matérn $3/2$ covariance.}\label{fig:sim_matern}
\end{figure}

\begin{figure}[H]
    \centering
    \begin{subfigure}{0.3\textwidth}
        \centering
        \includegraphics[width=\linewidth]{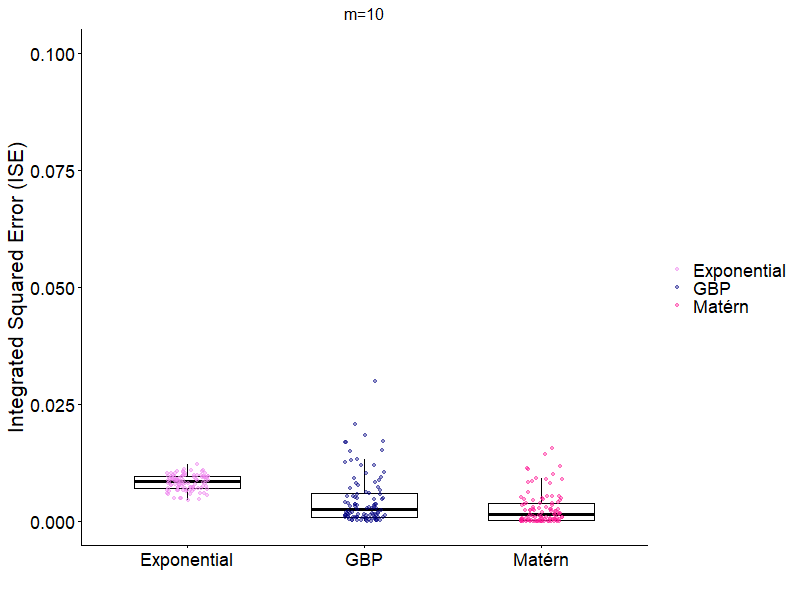}
        \caption{$n=100$}
    \end{subfigure}
    \hfill
    \begin{subfigure}{0.3\textwidth}
        \centering
        \includegraphics[width=\linewidth]{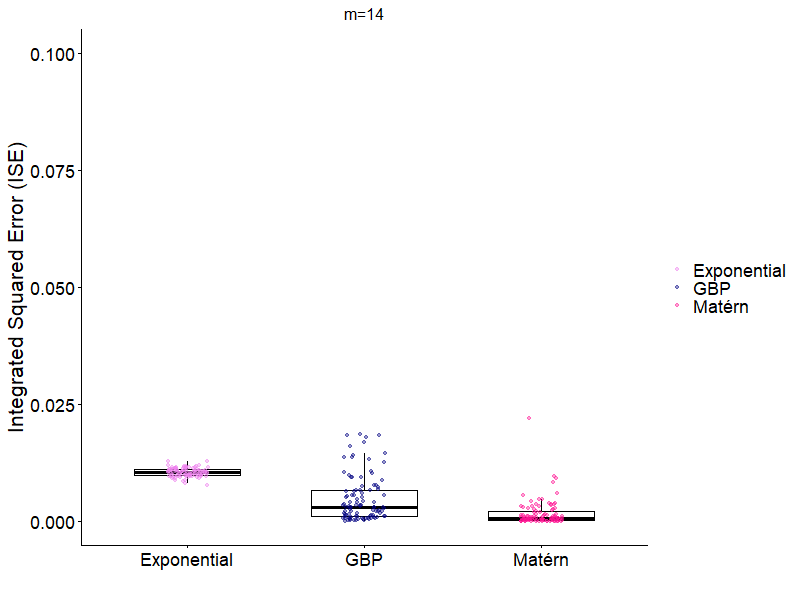}
        \caption{$n=200$}
    \end{subfigure}
    \hfill
    \begin{subfigure}{0.3\textwidth}
        \centering
        \includegraphics[width=\linewidth]{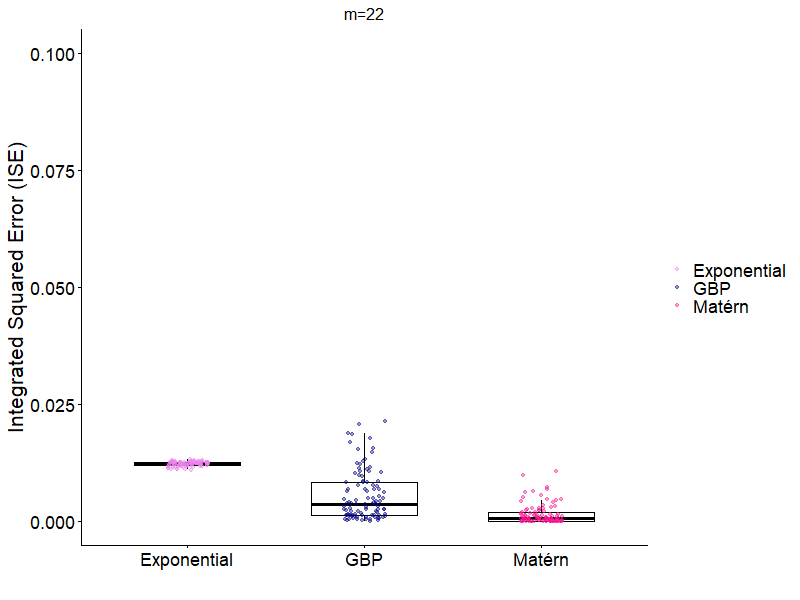}
        \caption{$n=500$}
    \end{subfigure}
    \caption{Boxplot of ISE: Matérn $3/2$ covariance case.}\label{fig:dist_matern}
\end{figure}

For all three sample sizes, the curve estimated by the GBP model was the one that most closely approximated the true theoretical curve, outperformed only by the Matérn function itself. In contrast, the exponential curve showed a less steep decline, deviating from the true pattern. This divergent behavior of the exponential function was theoretically expected due to its structural rigidity regarding differentiability at the origin ($d \to 0$). Since the true covariance structure (Matérn 3/2) describes a process with a higher degree of smoothness, the exponential model fails to adapt adequately and tends to inflate the spatial scale parameter. Consequently, the exponential curve decays more slowly with distance, overestimating the process's practical range. The GBP model, in turn, demonstrated sufficient flexibility to capture the observations generated by the Matérn model without distorting the tail behavior of the covariance function. 

\begin{figure}[htb]
    \centering
    
    \begin{subfigure}{\textwidth}
        \centering
        \includegraphics[height=5cm]{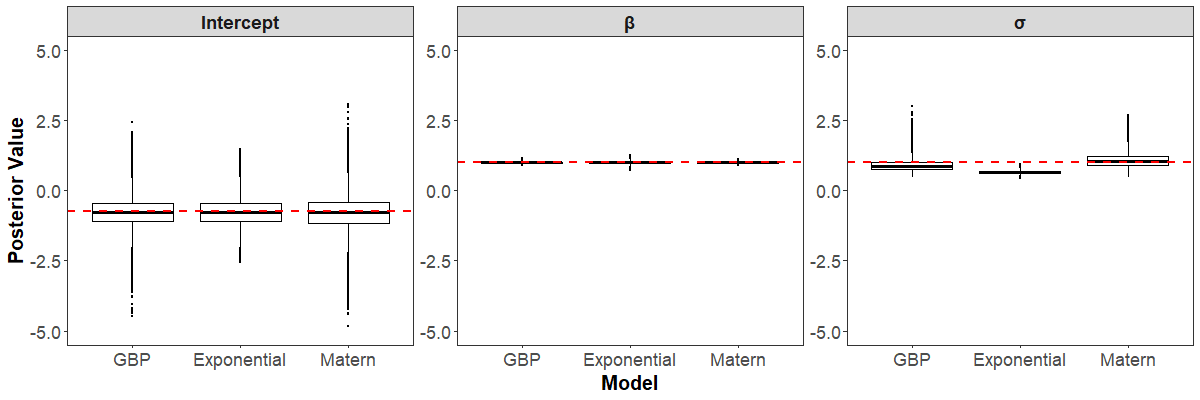}
        \caption{$n=200$.}
        \label{fig:param_matern200}
    \end{subfigure}
    
    
    \caption{Box plot of estimated parameters: Matérn case.}
    \label{fig:param_matern_all}
\end{figure}

Figure~\ref{fig:param_matern_all} summarizes the estimates obtained across the 100 simulations for the regression coefficients and the scale parameter. The results for $n=200$ indicate the convergence of the posterior medians to the true values for all models regarding the parameters $\beta_0$ and $\beta_1$. Regarding the parameter $\sigma$, the Matérn model recovered the exact value at the median. The GBP model showed a median close to the true value, whereas the exponential function—despite exhibiting lower dispersion—produced the fit that deviated most from the true parameter. For $n=100$ and $n=500$, the plots are supressed, but the conclusion are similar to the ones in Section~\ref{sec:exp:sim} and sample size $n=200$.
 The results of this scenario reinforce that GBP-based models are viable and highly flexible alternatives for fitting geostatistical models when the true covariance of the data is unknown.

\subsection{Scenario with power exponential covariance}

In the third scenario, data were generated using a power exponential covariance function, illustrating a situation where there is a misspecifcation of the parametric covariance structure in the fitted models. The goal is to evaluate the robustness of the proposed approach and the sensitivity of classical parametric functions. The power exponential family was selected to directly control the degree of surface smoothness via the shape parameter $p$. Two extremes of smoothness were investigated: ($a$) low smoothness ($p=0.5$), representing a highly irregular and non-differentiable process; and ($b$) high smoothness ($p=1.8$), generating a continuous and smooth spatial surface. These formulations allow for testing the adaptability of the GBP structure to different decay signatures. Figure~\ref{fig:sim_power05} displays the medians of the estimated spatial correlation curves for the high-irregularity scenario ($p=0.5$), and Figure \ref{fig:dist_power05} presents the boxplot of the ISE.

\begin{figure}[htb]
  \centering
    \begin{subfigure}{0.3\textwidth}
        \centering
        \includegraphics[width=\linewidth]{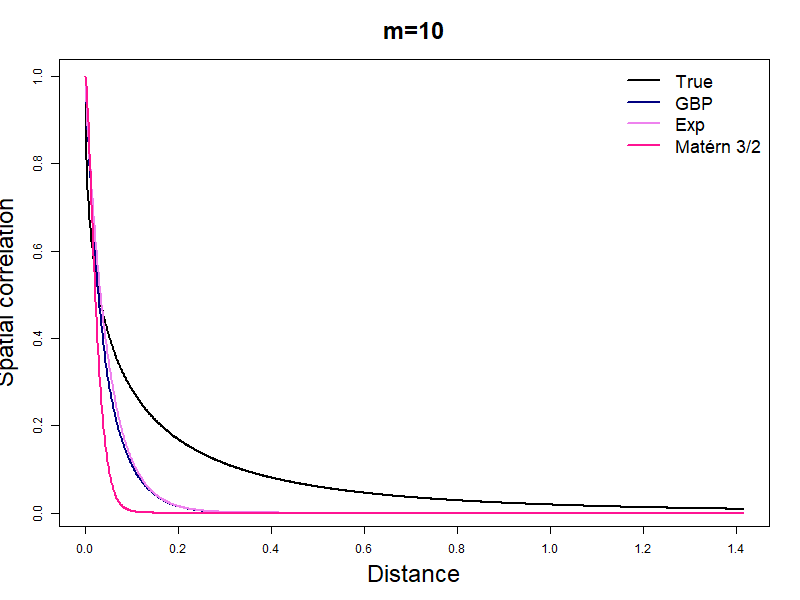}
        \caption{$n=100$}
    \end{subfigure}
    \hfill
    \begin{subfigure}{0.3\textwidth}
        \centering
        \includegraphics[width=\linewidth]{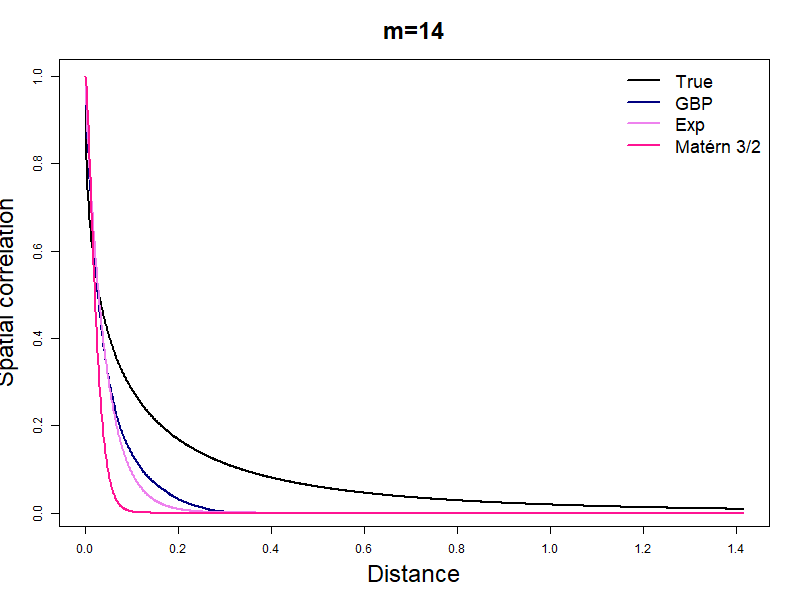}
        \caption{$n=200$}
    \end{subfigure}
    \hfill
    \begin{subfigure}{0.3\textwidth}
        \centering
        \includegraphics[width=\linewidth]{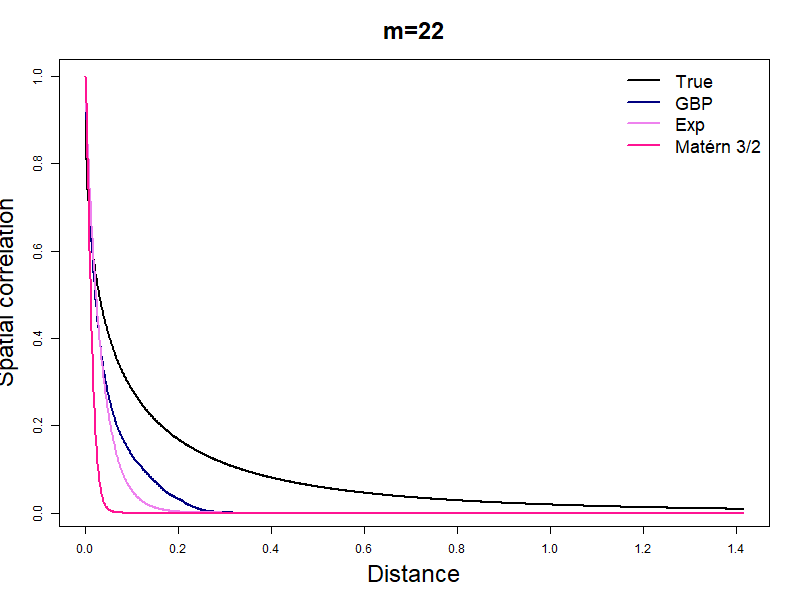}
        \caption{$n=500$}
    \end{subfigure}
    \caption{Spatial correlation curves:power exponential ($p=0.5$).}\label{fig:sim_power05}
\end{figure}

\begin{figure}[htb]
    \centering
    \begin{subfigure}{0.3\textwidth}
        \centering
        \includegraphics[width=\linewidth]{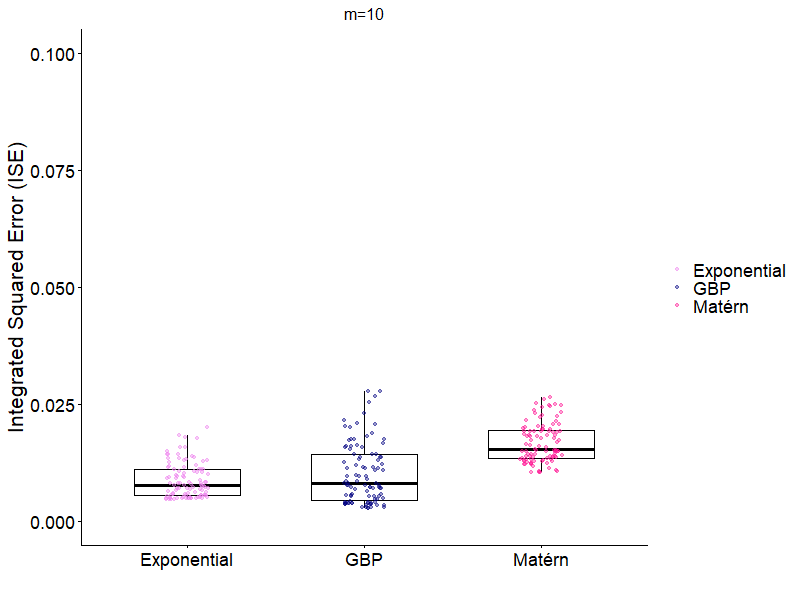}
        \caption{$n=100$}
    \end{subfigure}
    \hfill
    \begin{subfigure}{0.3\textwidth}
        \centering
        \includegraphics[width=\linewidth]{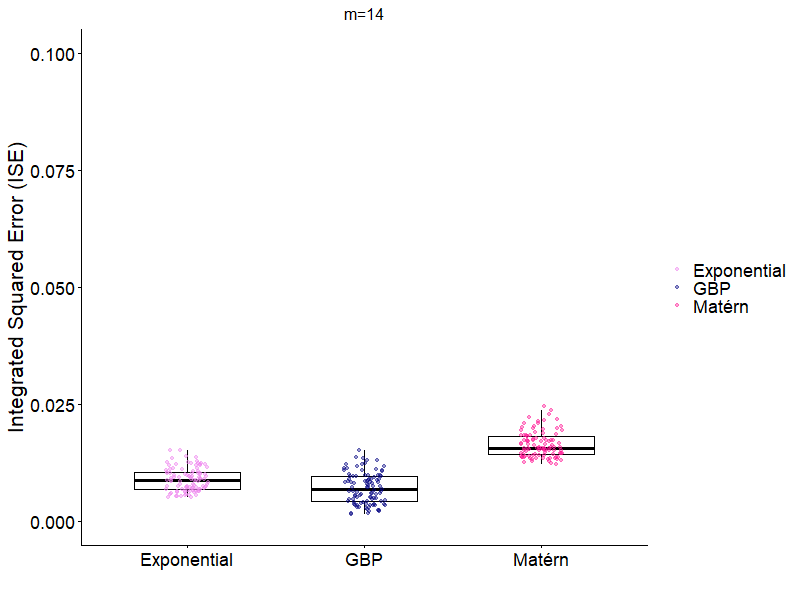}
        \caption{$n=200$}
    \end{subfigure}
    \hfill
    \begin{subfigure}{0.3\textwidth}
        \centering
        \includegraphics[width=\linewidth]{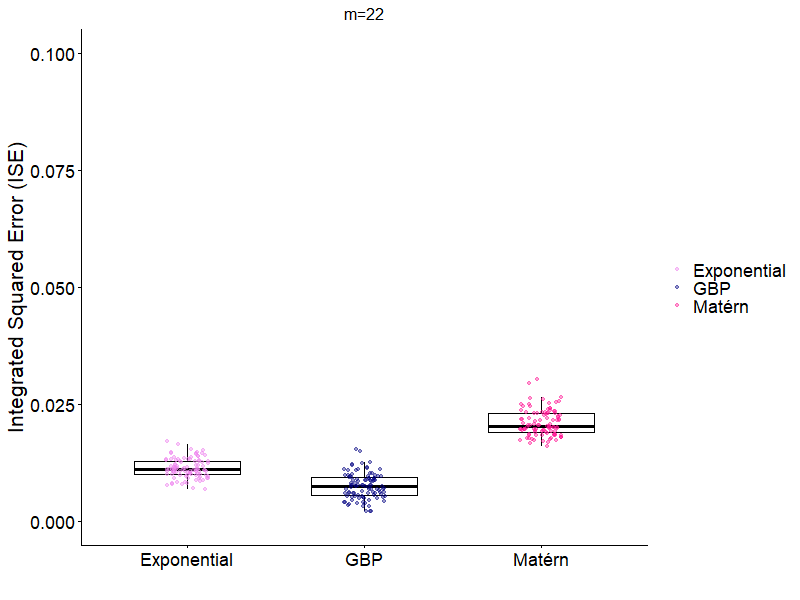}
        \caption{$n=500$}
    \end{subfigure}
    \caption{Boxplot of ISE: power exponential covariance case $(p=0.5)$.}\label{fig:dist_power05}
\end{figure}

Analysis of the estimated surface curves across the three sample sizes reveals that the GBP model was the most consistent among those evaluated. Although it does not perfectly capture the full smoothness structure at the origin for small samples, GBP outperformed the Exponential and Matérn models in all aspects regarding the approximation of the true curve. Figure~\ref{fig:param_power05_all} presents boxplots of the parameter estimates.

\begin{figure}[H]
    \centering
    \begin{subfigure}{\textwidth}
        \centering
        \includegraphics[height=5cm]{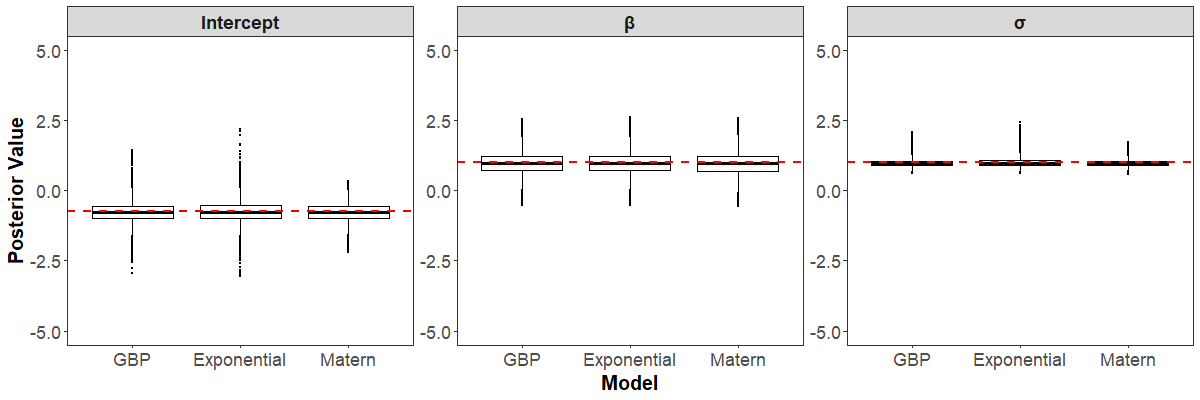}
        \caption{$n=100$.}
        \label{fig:param_power5100}
    \end{subfigure}
    \vspace{0.3cm} 
    
    \begin{subfigure}{\textwidth}
        \centering
        \includegraphics[height=5cm]{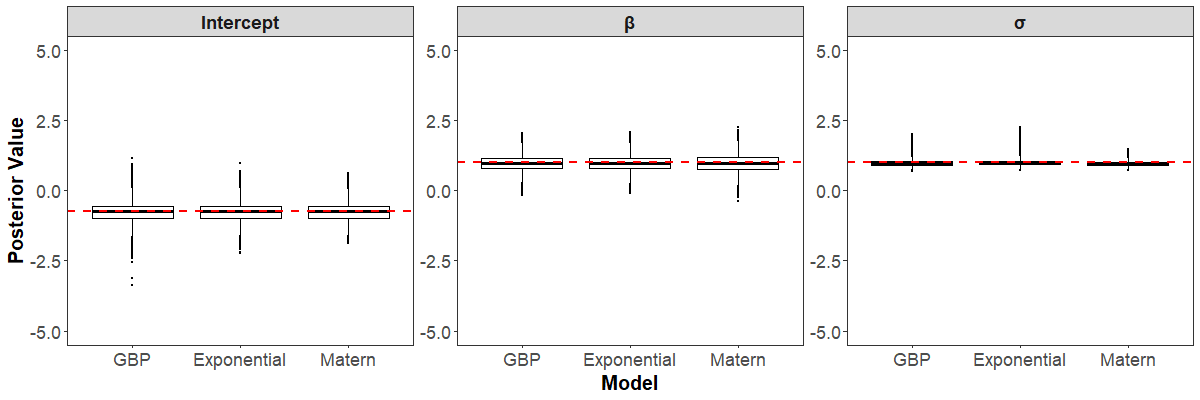}
        \caption{$n=200$.}
        \label{fig:param_power5200}
    \end{subfigure}
    \vspace{0.3cm}
    
    \begin{subfigure}{\textwidth}
        \centering
        \includegraphics[height=5cm]{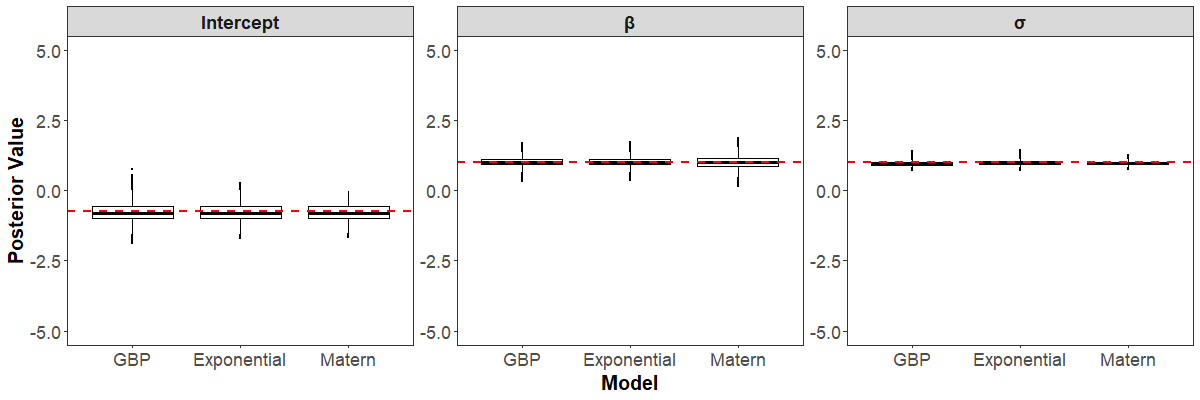}
        \caption{$n=500$.}
        \label{fig:param_power5500}
    \end{subfigure}
    
    \caption{Box plot of estimated parameters: power exponential ($p =1.8$).}
    \label{fig:param_power05_all}
\end{figure}

For $n=100$ and $n=200$, the models yielded similar results for the regression coefficients, with the posterior median converging to the true parameter.   For the scale parameter $\sigma$, all fitted models show little variability around the true value. However, based on the estimated correlation curves, this case supports the hypothesis that the GBP model represents a robust and reliable semiparametric alternative for spatial modeling under incorrect specifications.

We also consider data generated from the exponential power function ($b$) with a shape parameter of $p=1.8$. Figure~\ref{fig:sim_power18} illustrates the medians of the estimated spatial projection curves, and Figure \ref{fig:dist_power18} presents the boxplot of the ISE.

\begin{figure}[htb]
    \centering
    \begin{subfigure}{0.3\textwidth}
        \centering
        \includegraphics[width=\linewidth]{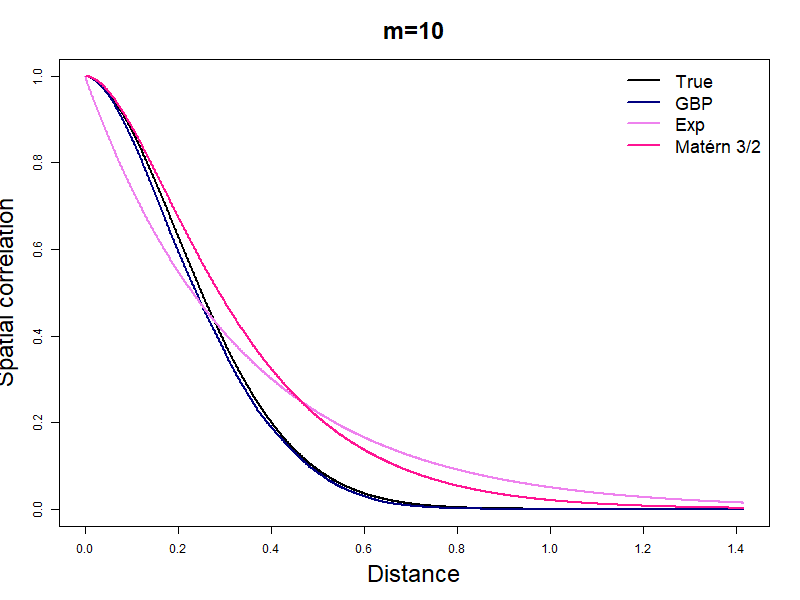}
        \caption{$n=100$}
    \end{subfigure}
    \hfill
    \begin{subfigure}{0.3\textwidth}
        \centering
        \includegraphics[width=\linewidth]{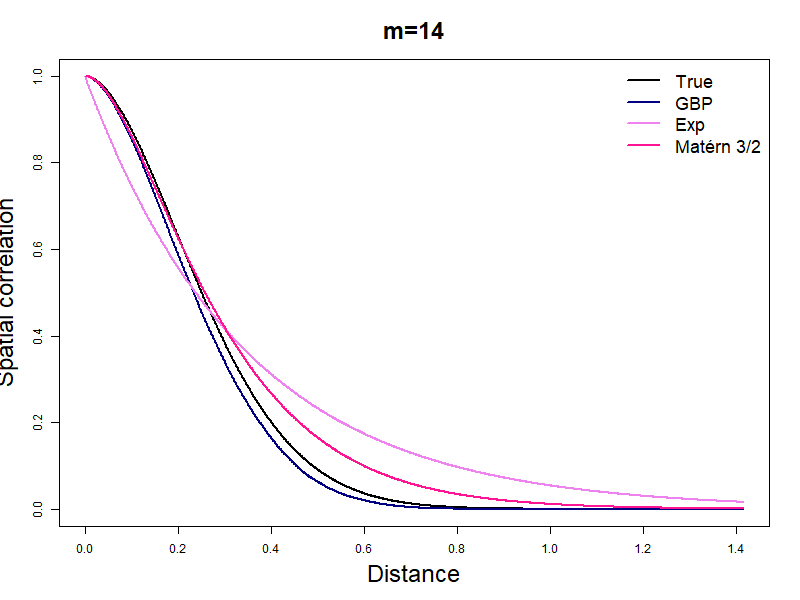}
        \caption{$n=200$}
    \end{subfigure}
    \hfill
    \begin{subfigure}{0.3\textwidth}
        \centering
        \includegraphics[width=\linewidth]{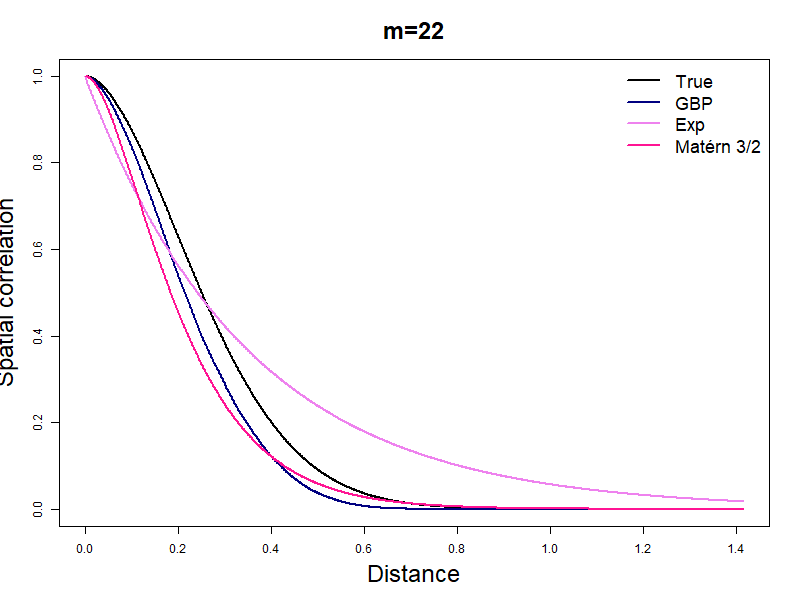}
        \caption{$n=500$}
  \end{subfigure}
    \caption{Spatial correlation curves: power exponential ($p=1.8$).}\label{fig:sim_power18}
\end{figure}
\begin{figure}[htb]
    \centering
    \begin{subfigure}{0.3\textwidth}
        \centering
        \includegraphics[width=\linewidth]{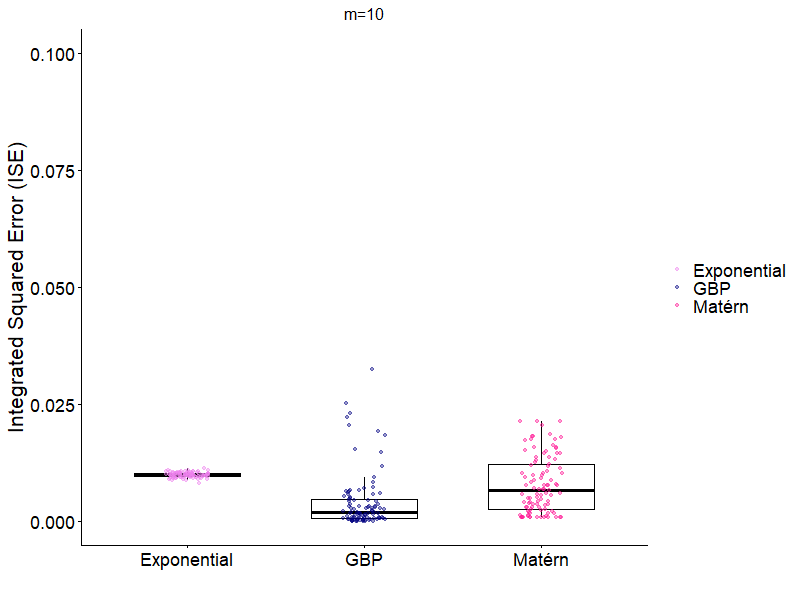}
        \caption{$n=100$}
    \end{subfigure}
    \hfill
    \begin{subfigure}{0.3\textwidth}
        \centering
        \includegraphics[width=\linewidth]{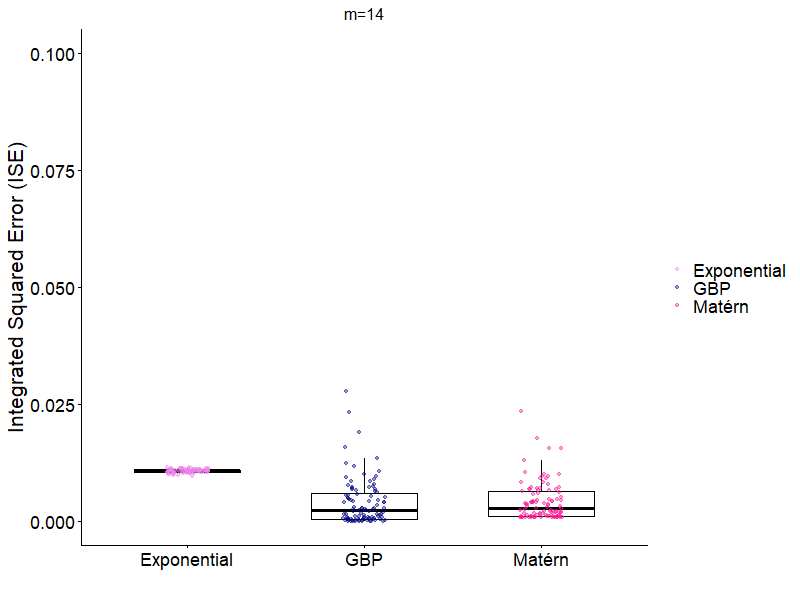}
        \caption{$n=200$}
    \end{subfigure}
    \hfill
    \begin{subfigure}{0.3\textwidth}
        \centering
        \includegraphics[width=\linewidth]{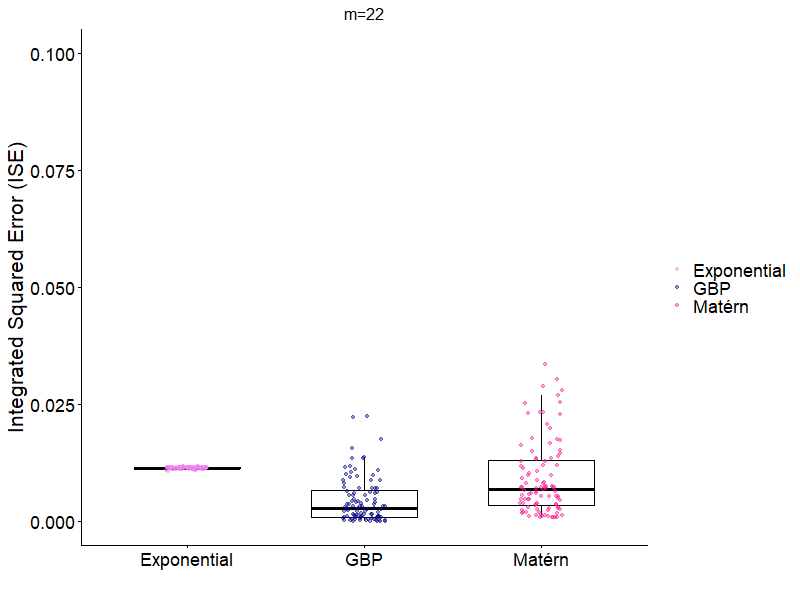}
        \caption{$n=500$}
    \end{subfigure}
    \caption{Boxplot of ISE: power exponential covariance case $(p=1.8)$.}\label{fig:dist_power18}
\end{figure}

Unlike the previous scenario ($p=0.5$), this theoretical observation curve does not exhibit an abrupt loss of continuity at the origin. Consequently, the fitted models aim to better understand and adapt to the behavior of the generating process. For $n=200$, the curve estimated by the GBP model was the one that most closely approximated the true curve across all observed distances. For $n=100$ and $n=500$, the model underperformed compared to the Matérn function within a small distance range. Despite minor fluctuations, the GBP model generally exhibits the covariance structure that best captures the spatial dependence of the data, followed by the Matérn and Exponential fits.

\begin{figure}[htb]
    \centering
    
    \begin{subfigure}{\textwidth}
        \centering
        \includegraphics[height=5cm]{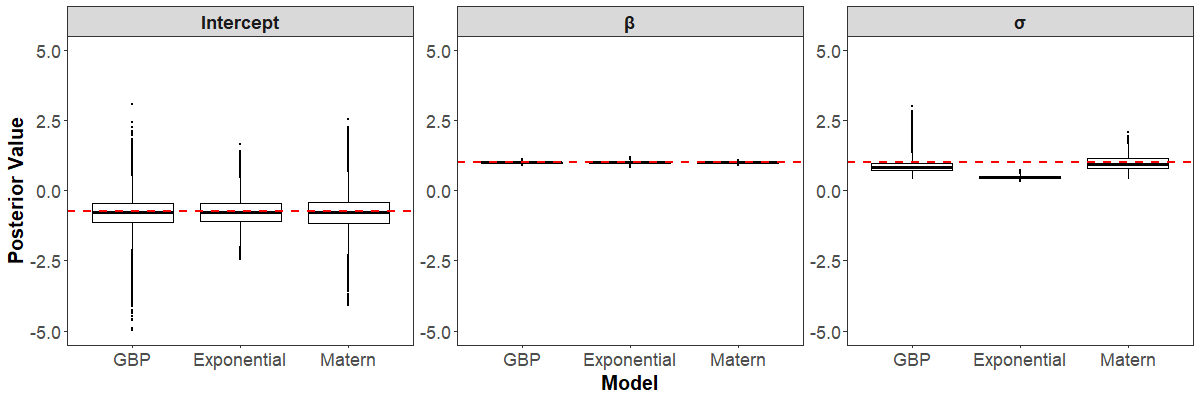}
        \caption{$n=200$.}
        \label{fig:param_power18200}
    \end{subfigure}
    
    
    \caption{Box plot of estimated parameters: power exponential ($p =1.8$).}
    \label{fig:param_power18_all}
\end{figure}

Figure~\ref{fig:param_power18_all} show the results for $n=200$ and indicates that, as in the previous scenarios, the fitted models satisfactorily recovered the regression coefficients $\beta_0$ and $\beta_1$, with posterior medians converging to the true values. The slight inconsistency observed in the spatial scale parameter $\sigma$ are also observed for $n=100$ and $n=500$ and may reflect asymptotic convergence difficulties regarding the true value. 

Certain methodological factors must be considered when jointly interpreting these findings. First, the GBP model may be sensitive to the choice of polynomial degree ($m$). However, even with a suboptimal degree choice, the results provide robust evidence that the Bernstein polynomial-based approach establishes itself as a reliable and flexible alternative for geostatistical modeling when the true covariance structure of the data is unknown.


\section{Application to real data}\label{sec:aplicacao}

The evaluation of the proposed model is based on data from the 2025 North American Nesting Bird Census (BBS). The BBS is a long-term annual survey, conducted since 1966, with the sole exception of 2020 due to COVID-19 pandemic restrictions, and contains point counts of over 700 bird species in North America. Data collection occurs primarily in June, during the peak breeding season, along thousands of randomly selected geographic routes in the United States and Canada. At each designated collection point on a route, a qualified bird identification scientist conducts a 3-minute point count, recording all species detected within a 400-meter radius. The complete dataset, along with comprehensive documentation on methodologies and variables, is publicly available at \url{https://www.sciencebase.gov/catalog/item/691cfb53d4be021d1d89b482}.
Latitude acts as a ``summary" of different macroecological characteristics, correlating directly with local temperature, reproductive events, and the behavioral transition of its population, such as variations from strict migrants at higher latitudes to permanent residents at lower latitudes \citep{jahn2020asymmetric}. These factors reinforce the need to model spatial covariance in a flexible and semiparametric way.

For model fitting and analysis of the American Robin distribution, a total sample of $n_{total} = 534$ routes observed across the 18 reported states will be considered. Of these, $n=500$ will be used for model fitting and $n_{pred} =34$ for prediction analysis. The sample data used meet the BBS quality criteria, i.e., they were found under ideal conditions of climate, data, time, and route completion; a route was randomly selected, and the observer followed the official sampling protocol.

\subsection{Modeling Approaches}

The application's interest lies in analyzing the abundance and distribution of the American Robin. The process of specifying, adjusting, and selecting the predictive models was conducted in two methodological stages. In the first stage, the focus was on modeling the average behavior of the process, as in the model presented in Section \ref{sec:gshm}, through a spatial latent variable with a distribution characterized by the proposed covariance function. Thus,
\begin{equation*}
\label{eq:spatial_data_model_gen}
Y(\mathbf{s}_i) \mid \theta(\mathbf{s}_i), \phi \sim f\left(Y(\mathbf{s}_i) \mid \mu(\mathbf{s}_i), \phi\right),
\end{equation*}
where $Y(\mathbf{s}_i)$ denotes the total number of birds observed on the route $\mathbf{s}_i$, and $\phi$ is the dispersion or shape parameter associated with the chosen distribution family $f(\cdot)$, $i=1, \ldots, n$. 

Under this fixed mean structure, the stochastic behavior of the data was investigated by testing different probabilistic hypotheses for the distribution of the response variable, since ecological count data frequently exhibit overdispersion behavior \citep{bolker2008ecological, ohara2010do}. Thus, the models were fitted under the Poisson, negative binomial, and Bell distributions. Therefore, for each location $i = 1, \dots, n$, the count variable $Y(\mathbf{s}_i)=Y_i$ is conditionally modeled by a discrete distribution:
\begin{align*}
Y_i \mid \mu_i, \psi &\sim f(\mu_i, \phi),\\
\mu_i &= \beta_0^\star +  X_i^\top \boldsymbol{\beta^*}, \\
\frac{1}{\sqrt{\psi}} &\sim \text{Gamma}(0.1, 0.1),
\end{align*}
where $\mu_i = E(Y_i)$ represents the expected average abundance, $\psi$ is the dispersion parameter in the negative binomial model or equal to 1 in Poisson and the Bell distributions. The $\beta_0^*$ and $\boldsymbol{\beta^*}$ parameters denote that the prior distributions have been placed in the standardized coefficients ($\beta_0$ and $\boldsymbol{\beta}$). Non-informative prior distributions were assigned to the hyperparameters, reflecting the absence of substantial prior knowledge; they were selected according to the traditional prior selection recommendations presented in \cite{rstan2026}.
The prior assigned are the same used in Section~\ref{sec:sim}.

In the parameterization used for the Negative Binomial, where $Var(Y_i) = \mu_i + \mu_i^2/\psi$, the dispersion parameter $\psi$ exhibits an asymptotic and highly non-linear relationship with the model variance. Values of $\psi$ close to zero generate explosions in variance, while large values approximate the distribution to the Poisson. Assigning a prior distribution directly to $\psi$ negatively impacts the sampling of the HMC algorithm, generating numerical instability (divergent transitions). To overcome this limitation, the prior is allocated to a reparameterized $1/\sqrt{\psi}$ for two reasons: the first is the linear interpretation, since in the formulation of the Negative Binomial as a Poisson-Gamma mixture, this reparameterization directly represents the standard deviation of the overdispersion multiplier effect, allowing for the choice of priors on a more intuitive scale; and second, by numerical stabilization, given that the transformation linearizes the critical region close to zero and compacts the long tail when $\psi \to \infty$, which optimizes the calculation of gradients by Stan and improves the convergence of MCMC \citep{rstan2026}.

The structured linear predictor, using the logarithmic link function and the settling component (\textit{offset}), is defined by:
\begin{equation*}
\log(\mu_i) = \beta_0 + \mathbf{X}_{i} \boldsymbol{\beta} + \theta_i + \text{offset}_i,
\end{equation*}
where $\beta_0$ represents the global intercept, $\mathbf{X} = ( X_1, X_2)$ is the covariate matrix, latitude and latitude squared \citep{zaida}, with coefficients $\boldsymbol{\beta}= (\beta_1, \beta_2)$ and  $\boldsymbol{\theta} = (\theta_1, \dots, \theta_n)^{\top}$ is the residual spatial dependence captured through a Gaussian process:
\begin{equation*}
\boldsymbol{\theta} \sim GP\left(\mathbf{0}, \boldsymbol{\Sigma}\right),
\end{equation*}
in which $\boldsymbol{\Sigma}$ is covariance a matrix.  

Initially, which distributions were combined with the latent stochastic spatial component structured by the GBP, also evaluating the impact of the polynomial degree on the smoothing of the spatial surface (comparisons between $m=8$ and $m=23$). These values were chosen according to what is suggested by \cite{osman2012nonparametric} with $m \approx \sqrt{n_{total}}$ or $m \approx n_{total}^{1/3}$ \citep{mclain2009estimation}.

After identifying the best predictive model within the proposed response distributions, the second stage was carried out, which consisted of comparative validation with the parametric model. The best-performing model was directly compared with its competing version using the classical exponential theoretical covariance function for the latent process. Decisions on choosing the best model were based on out-of-sample predictive capacity metrics, evaluating the balance between minimizing mean forecast errors (RMSE and MAE) and the properties of the credibility/predictive intervals, considering a trimmed mean of 2.5\%. Specifically, priority was given to the approach that demonstrated the best probability of nominal coverage (CP) (95\% target) associated with the highest predictive accuracy, measured by the smallest mean interval width (MPIW).

For the evaluation of predictive capacity, the complete dataset was divided into two disjoint subsets: the observed data vector, denoted by $\mathbf{Y}_{\text{obs}}$, with dimension $n = 500$, used for parameter fitting and estimation; and the validation data vector, denoted by $\mathbf{Y}_{\text{pred}}$, with dimension $n_{\text{pred}} = 34$, reserved exclusively for measuring the predictive power of the model in unsampled locations.

The parameters of all fitted models were estimated by \texttt{rstan}. Table ~\ref{tab:comparacao_modelos} shows the predictive performance results of the models adjusted in the first stage.

\begin{table}[H]
\centering
\caption{Comparison of the performance of adjusted predictive models (GBP).}
\label{tab:comparacao_modelos}
\begin{tabular}{llcccc}
\toprule
\textbf{Distribution} & \textbf{Spatial Covariance} & \textbf{RMSE} & \textbf{MAE} & \textbf{MPIW} & \textbf{CP} \\
\midrule
Bell           & GBP ($m=8$)   & 45.256 & 39.198 & 132.190 & 0.618 \\
Bell              & GBP ($m=23$)  & 44.252 & 38.140 & 127.602 & 0.618 \\
Negative binomial & GBP ($m=8$) & 32.566 & 21.835 & 112.195 & 0.882 \\
\rowcolor{gray!20}Negative binomial & GBP ($m=23$)  & \textbf{26.785} & \textbf{19.811} & \textbf{101.309} & \textbf{0.971} \\
Poisson           & GBP ($m=8$)   & 32.462 & 21.851 & 115.460 & 0.912\\
Poisson           & GBP ($m=23$)  & 32.412 & 21.722 & 113.870 & 0.912 \\

\bottomrule
\end{tabular}
\end{table}

The results in Table~\ref{tab:comparacao_modelos} show that the choice of the variable response distribution has a crucial impact on predictive capacity. Models fitted under the assumption of a negative binomial distribution outperformed approaches based on Poisson and Bell distributions in all of the evaluated metrics. This behavior is statistically justified by the presence of strong overdispersion in ecological species count data, a characteristic that is particularly modeled by the extra dispersion parameter of the negative binomial. Although the Bell distribution is an alternative for overdispersed data, it is not as flexible due to the deterministic relation of its variance, making it a more viable option for moderately overdispersed data.

Among the models evaluated, the one composed of the negative binomial distribution associated with the proposed spatial covariance function based on Bernstein Polynomials with $m=23$ (GBP, $m=23$) was the one that presented the best predictive approach. This model obtained the lowest validation errors, with a Mean Squared Error (RMSE) of $26.785$ and a Mean Absolute Error (MAE) of $19.811$. It achieved a Coverage Probability of 97.1\%, guaranteeing statistical reliability, in addition to the lowest Mean Predictive Interval Width (MPIW = 101.309) observed. Thus, GBP ($m=23$) was the model chosen for the next step.
\begin{table}[htb]
\centering
\caption{Final comparison of the performance of adjusted predictive models.}
\label{tab:comparacao_modelos2}
\begin{tabular}{llcccc}
\toprule
\textbf{Distribution} & \textbf{Spatial Covariance} & \textbf{RMSE} & \textbf{MAE} & \textbf{MPIW} & \textbf{CP} \\
\midrule
Negative binomial & Exponential & 26.649 & 19.491 & 102.749 & \textbf{0.971} \\
\rowcolor{gray!20}Negative binomial & GBP ($m=23$)  & 26.785 & 19.811 & 101.309 & \textbf{0.971} \\
Negative binomial &             Mat\'ern $3/2$ & \textbf{26.131} & \textbf{19.415} & \textbf{100.785} & \textbf{0.971} \\
\bottomrule
\end{tabular}
\end{table}

Table \ref{tab:comparacao_modelos2} presents a comparison of the performance of the GBP and parametric (exponential and Matérn $3/2$ covariances, two of the most common in the literature) predictive models. The results reveal a pattern of equivalence and complementarity between the latent structures. The Matérn covariance showed a marginal advantage in terms of average prediction errors, with an RMSE of 26.131 and MAE of 19.415, compared to the RMSE of 26.785 and MAE of 19.811 obtained with the GBP structure. All methodologies yielded an identical coverage probability (CP) of 97.1\%. The Matérn structure also achieved the narrowest Mean Prediction Interval Width (MPIW = 101.309 vs. 101.309 for the GBP model) while maintaining the same optimal coverage level. 

\begin{table}[htb]
\centering
\caption{Posterior descriptive statistics of the structural parameters.}
\label{tab:spatial_model_summary}
\small
\begin{tabular}{lrrrrrrrr}
\toprule
\textbf{Parameter} & \multicolumn{4}{r}{GBP-NB} & \multicolumn{4}{r}
{Mat\'ern 3/2-NB} \\
\cmidrule{2-5}  \cmidrule{6-9} 
& \textbf{Mean} & \textbf{SE Mean} & \textbf{SD} & \textbf{95\% CI} & \textbf{Mean} & \textbf{SE Mean} & \textbf{SD} & \textbf{95\% CI} \\
\midrule
$\beta_0$ & -25.656 & 1.614 & 15.501 & [-54.763, 4.401] & -20.703 & 0.817 & 15.082 & [-49.306, 8.562]\\
$\beta_1$ &   1.314 & 0.080 &  0.776 & [ -0.191, 2.758] &   1.059 & 0.043 &  0.762 & [ -0.430, 2.525]\\
$\beta_2$ &  -0.015 & 0.001 &  0.010 & [ -0.033, 0.004] &  -0.011 & 0.001 &  0.010 & [ -0.030, 0.007]\\
$\sigma^2$  &   0.401 & 0.012 &  0.161 & [  0.195, 0.860] &   0.695 & 0.014 &  0.180 & [  0.459, 1.143]\\
$\psi$ & 3.774 & 0.053 &  0.438 &   [3.089, 4.823] &   3.460 & 0.011 &  0.269 & [2.963, 4.026]\\
\bottomrule
\end{tabular}
\end{table}

Table \ref{tab:spatial_model_summary} presents the posterior estimates of the parameter descriptive statistics for the GBP and Matérn models. It is observed that the 95\% credible intervals for the parameters $\beta_1$ (95\% CI: $[-0.191; 2.758]$ for the GBP and 95\% CI: $[-0.430; 2.525]$ for the Matérn) and $\beta_2$ (95\% CI: $[-0.033; 0.004]$ for the GBP and 95\% CI: $[-0.030; 0.007]$ for the Matérn) include the value zero; that is, there is no statistical evidence that these covariates exert an effect on the response variable at the significance level considered. It is important to note that, although the coefficients $\beta_1$ and $\beta_2$ are statistically significant in the model proposed by \cite{zaida}, that fit assumes the data follow a Poisson distribution. This difference can be explained by the inclusion of the negative binomial distribution's dispersion parameter, which controls for overdispersion and alters the inference result. Failure to adequately account for overdispersion typically leads to an underestimation of standard errors and an inflation of statistical significance in tests where none exists.

The spatial variance parameter $\sigma^2$ showed a posterior mean of 0.401 (95\% CI: $[0.195; 0.860]$) for the GBP and 0.695 (95\% CI: $[0.459; 1.143]$) for the Matérn, indicating significant and well-estimated spatial variance in both models. Finally, the parameter $\psi$, associated with the negative binomial distribution structure, showed a mean of 3.774 with a strictly positive credible interval ($[3.089; 4.823]$) for the GBP fit, and a mean of 3.460 (95\% CI: $[2.963; 4.026]$) for the Matérn fit. These results confirm the presence of overdispersion in the data and highlight the suitability of the negative binomial model. Furthermore, they show that although the Matérn covariance yields slightly better prediction performance over the metrics, the gains are minor and there is an agreement between the estimated parameters between the two models. Therefore, we conclude that the GBP structure as a viable alternative. Thus, we further investigate its performance. The convergence analysis was successful and is presented in Appendix \ref{app:conv}.

\begin{figure}[H]
\begin{center}
\includegraphics[height=6cm]{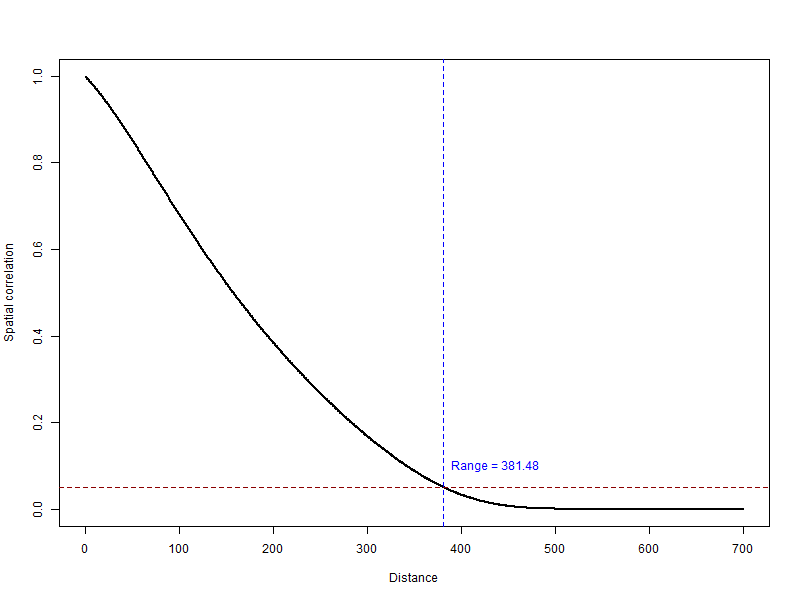}
\caption{Spatial correlation function as a function of geographic distance, highlighting the range (range = 381.48) estimated by the GBP-Negative binomial model.}\label{fig:range_nb}
\end{center}
\end{figure}

Figure \ref{fig:range_nb} shows the spatial correlation curve estimated by the model. As expected, the spatial correlation starts at $1.0$ for a zero distance (indicating that the variability at a given point is perfectly self-explained) and decreases continuously as the physical distance expands. The behavior of this curve allows us to estimate the range parameter, whose estimated value was 381.48 km, that is, the spatial dependence acts significantly up to a radius of approximately 381.48 km. For distances greater than this value, the spatial correlation converges asymptotically to zero (below the threshold of $0.05$, represented by the red horizontal dashed line). Next, we evaluated its ability to reproduce the observed patterns and its efficiency in extrapolating these estimates to unsampled regions.

Mapping the posterior mean of the transformed linear predictor, $\boldsymbol{\mu} = \exp(\boldsymbol{\eta})$ with $\boldsymbol{\eta} = \beta_0 + \mathbf{X}\beta + \boldsymbol{\theta}$ (Figure \ref{fig:theta_predic}), reveals a continuous and smoothed surface, highlighting the model's ability to capture the spatial dependence of the phenomenon. The central regions of the map, particularly Pennsylvania and New York, concentrate the highest expected values for the abundance. Conversely, the lowest estimates are located in the southern portions of the study area and in the far northeast (Maine).
\begin{figure}[H]
\begin{center}
\includegraphics[height=6cm]{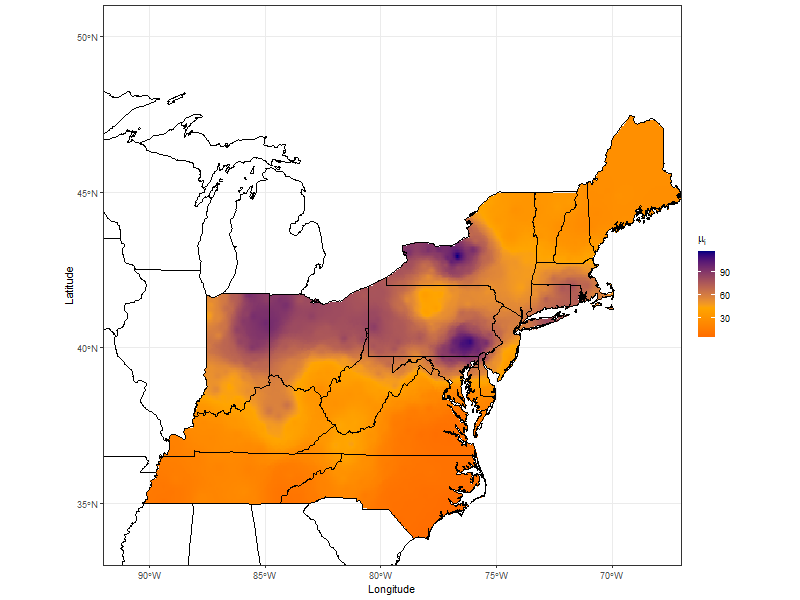}
\caption{Posterior mean of the expected abundance ($\boldsymbol{\mu}$) for the American Robin data. The continuous surface highlights regional patterns of overabundance (shades of blue) and underabundance (shades of orange) in the eastern United States.}\label{fig:theta_predic}
\end{center}
\end{figure}

Out-of-sample predictive performance was evaluated using a cross-validation design with $n_{\text{pred}} = 34$ independent test routes. The overall fit of the estimates is illustrated in Figure \ref{fig:validacao}, where we present the individual predictive performance of each test route in the validation set. Red triangles indicate the actual observed counts, empty blue circles represent point projections, and light blue vertical bars illustrate the 95\% predictive credibility intervals, highlighting the model's CP calibration.

\begin{figure}[H]
\begin{center}
\includegraphics[height=6cm]{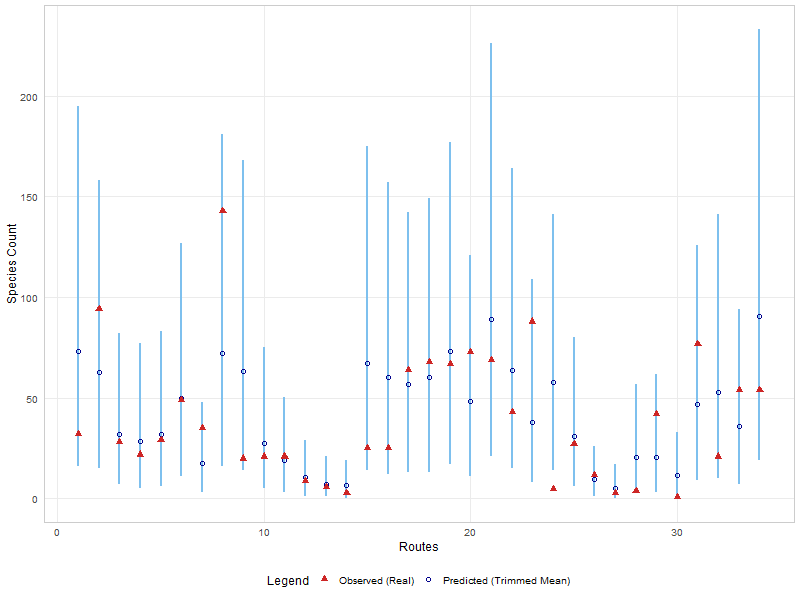}
	\caption{Individual predictive performance for each test route in the out-of-sample validation set.}\label{fig:validacao}
\end{center}
\end{figure}

It is observed that the CP of the 95\% credibility interval was able to encompass almost all of the real data (red triangles). Even in scenarios where the point estimate diverged from the observed value, the uncertainty associated with the predictive distribution was calibrated satisfactorily, without underestimating the projection errors.

We concluded that the model accurately mapped the spatial patterns of the American robin species in the eastern United States, capturing a spatial dependence that extends over a radius of 381.48 km. Furthermore, the results allow us to conclude that integrating the negative binomial counting structure with the Gaussian process using BP resulted in a statistically robust, well-calibrated model with good generalization power.


\section{Conclusion}\label{sec:conclusao}

In this study we introduced a robust semiparametric approach applicable to any class of models incorporating latent spatial effects in the context of georeferenced data, structuring the flexibility of the covariance function based on Bernstein polynomials (Geo-Bernstein Polynomial - GBP). The proposed methodology was evaluated through an extensive simulation study and a practical application using real data on the abundance of the American Robin (\textit{Turdus migratorius}), obtained from the North American Breeding Bird Survey (BBS).

Simulation results demonstrated that, under various data-generating scenarios (including classic structures such as the Exponential and Matérn $3/2$ models), the GBP model exhibited a remarkable ability to approximate theoretical correlation curves. The model proved to be a reliable and adaptable alternative for capturing varying degrees of smoothness or irregularity in the spatial surface when the true covariance structure of the data is unknown. In the practical application, modeling the American Robin's distribution highlighted the importance of simultaneously accommodating data overdispersion and latent spatial dependence. The model structured with the Negative Binomial distribution and GBP covariance ($m=23$) outperformed approaches based on Poisson and Bell distributions across all out-of-sample cross-validation metrics. In a direct comparison with the traditional parametric model (exponential covariance), the GBP model yielded satisfactory results; both structures produced equivalent point predictions and achieved the same Nominal Coverage Probability (CP = 97.1\%), but the proposed model generated narrower prediction intervals (lower MPIW). This property indicates that the flexibility of Bernstein polynomials allows for a reduction in the width of uncertainty bands without compromising model calibration.

 We highlight that the model yielded satisfactory results and that its formulation within the Bayesian paradigm allows for the accommodation of highly nonlinear and complex structures through hierarchical modeling, making it applicable to both Gaussian and non-Gaussian data. 
 
Future work will involve examining the conditions under which the covariance matrix resulting from the GBP model is positive definite, as well as conducting comparative studies with other semiparametric models established in the literature. Furthermore, the aim is to extend the methodology to other contexts, such as spatial survival models, and to develop an R package to disseminate the proposal and facilitate its application by researchers across various fields.

\bibliographystyle{chicago}  
\bibliography{references}  

\begin{thebibliography}{}

\bibitem[\protect\citeauthoryear{Banerjee}{Banerjee}{2015}]{banerjee2015book}
Banerjee, Sudipto;~Carlin, B. P. G. A.~E. (2015).
\newblock {\em Hierarchical Modeling and Analysis for Spatial Data, Second
  Edition\/} (2ed. ed.).
\newblock Chapman \& Hall/CRC Monographs on Statistics \& Applied Probability.
  CRC Press.

\bibitem[\protect\citeauthoryear{Bolker}{Bolker}{2008}]{bolker2008ecological}
Bolker, B.~M. (2008).
\newblock {\em Ecological models and data in R}.
\newblock Princeton: Princeton University Press.

\bibitem[\protect\citeauthoryear{Cressie}{Cressie}{2015}]{cressie1993}
Cressie, N. (2015).
\newblock {\em Statistics for Spatial Data}.
\newblock John Wiley \& Sons.

\bibitem[\protect\citeauthoryear{Diggle and Ribeiro}{Diggle and
  Ribeiro}{2007}]{diggle2007}
Diggle, P.~J. and P.~J. Ribeiro (2007).
\newblock {\em Model-based Geostatistics}.
\newblock Springer Series in Statistics. New York: Springer.

\bibitem[\protect\citeauthoryear{Gelfand, Kottas, and MacEachern}{Gelfand
  et~al.}{2005}]{gelfand2005spatial}
Gelfand, A.~E., A.~Kottas, and S.~N. MacEachern (2005).
\newblock Bayesian nonparametric spatial modeling with {D}irichlet process
  mixing.
\newblock {\em Journal of the American Statistical Association\/}~{\em
  100\/}(471), 1021--1035.

\bibitem[\protect\citeauthoryear{Genton and Gorsich}{Genton and
  Gorsich}{2002}]{genton2002nonparametric}
Genton, M.~G. and D.~J. Gorsich (2002).
\newblock Nonparametric variogram and covariogram estimation with
  {F}ourier--{B}essel matrices.
\newblock {\em Computational Statistics \& Data Analysis\/}~{\em 41\/}(1),
  47--57.

\bibitem[\protect\citeauthoryear{Ghosal}{Ghosal}{2001}]{ghosal2001posterior}
Ghosal, S. (2001).
\newblock Convergence rates for density estimation with bernstein polynomials.
\newblock {\em The Annals of Statistics\/}~{\em 29\/}(5), 1264--1280.

\bibitem[\protect\citeauthoryear{Ghosal and Van~der Vaart}{Ghosal and Van~der
  Vaart}{2017}]{ghosal2017fundamentals}
Ghosal, S. and A.~W. Van~der Vaart (2017).
\newblock {\em Fundamentals of nonparametric Bayesian inference}, Volume~44.
\newblock Cambridge University Press.

\bibitem[\protect\citeauthoryear{Hall, Fisher, and Hoffmann}{Hall
  et~al.}{1994}]{hall1994nonparametric}
Hall, P., N.~I. Fisher, and B.~Hoffmann (1994).
\newblock On the nonparametric estimation of covariance functions.
\newblock {\em The Annals of Statistics\/}, 2115--2134.

\bibitem[\protect\citeauthoryear{Hoffman and Gelman}{Hoffman and
  Gelman}{2014}]{hoffman2014nuts}
Hoffman, M.~D. and A.~Gelman (2014).
\newblock The no-u-turn sampler: adaptively setting path lengths in hamiltonian
  monte carlo.
\newblock {\em Journal of Machine Learning Research\/}~{\em 15\/}(1),
  1593--1623.

\bibitem[\protect\citeauthoryear{Jahn, Lerman, Phillips, Ryder, and
  Williams}{Jahn et~al.}{2019}]{jahn2020asymmetric}
Jahn, A.~E., S.~B. Lerman, L.~M. Phillips, T.~B. Ryder, and E.~J. Williams
  (2019).
\newblock First tracking of individual american robins (turdus migratorius)
  across seasons.
\newblock {\em The Wilson Journal of Ornithology\/}~{\em 131\/}(2), 356--359.

\bibitem[\protect\citeauthoryear{Lorentz}{Lorentz}{2012}]{book:5337}
Lorentz, G. (2012).
\newblock {\em Bernstein Polynomials}.
\newblock AMS Chelsea Publishing. American Mathematical Society.

\bibitem[\protect\citeauthoryear{Manchuk and Leuangthong}{Manchuk and
  Leuangthong}{2008}]{manchuk2008experimental}
Manchuk, J. and O.~Leuangthong (2008).
\newblock Experimental variogram calculation using bernsteing polynomials.
\newblock {\em Centre for Computational Geostatistics, Report\/}~{\em 8}.

\bibitem[\protect\citeauthoryear{McLain and Ghosh}{McLain and
  Ghosh}{2009}]{mclain2009estimation}
McLain, A.~C. and S.~K. Ghosh (2009).
\newblock Estimation of time transformation models with bernstein polynomials.
\newblock {\em Institute of Statistics Mimeo Series\/}~{\em 2625}.

\bibitem[\protect\citeauthoryear{Neal}{Neal}{2011}]{neal2011hmc}
Neal, R.~M. (2011).
\newblock Mcmc using hamiltonian dynamics.
\newblock In S.~Brooks, A.~Gelman, G.~Jones, and X.-L. Meng (Eds.), {\em
  Handbook of Markov Chain Monte Carlo}, pp.\  113--162. CRC Press.

\bibitem[\protect\citeauthoryear{O'Hara and Kotze}{O'Hara and
  Kotze}{2010}]{ohara2010do}
O'Hara, R.~B. and D.~J. Kotze (2010).
\newblock Do not log-transform count data.
\newblock {\em Methods in Ecology and Evolution\/}~{\em 1\/}(2), 118--122.

\bibitem[\protect\citeauthoryear{Osman and Ghosh}{Osman and
  Ghosh}{2012}]{osman2012nonparametric}
Osman, M. and S.~K. Ghosh (2012).
\newblock Nonparametric regression models for right-censored data using
  bernstein polynomials.
\newblock {\em Computational Statistics \& Data Analysis\/}~{\em 56\/}(3),
  559--573.

\bibitem[\protect\citeauthoryear{Quiroz, Prates, Gonzales, and Rue}{Quiroz
  et~al.}{2026}]{zaida}
Quiroz, Z., M.~Prates, C.~Gonzales, and H.~Rue (2026).
\newblock Bayesian inference of generalized linear blocknngp model of
  coregionalization through inla.
\newblock Submitted for publication.

\bibitem[\protect\citeauthoryear{{R Core Team}}{{R Core Team}}{2024}]{R}
{R Core Team} (2024).
\newblock {\em R: A Language and Environment for Statistical Computing}.
\newblock Vienna, Austria: R Foundation for Statistical Computing.

\bibitem[\protect\citeauthoryear{{Stan Development Team}}{{Stan Development
  Team}}{2026}]{rstan2026}
{Stan Development Team} (2026).
\newblock {\em {RStan}: the {R} interface to {Stan}}.
\newblock R package version 2.32.14.

\bibitem[\protect\citeauthoryear{Stein}{Stein}{1999}]{stein1999interpolation}
Stein, M.~L. (1999).
\newblock {\em Interpolation of spatial data: some theory for kriging}.
\newblock Springer Science \& Business Media.

\bibitem[\protect\citeauthoryear{Tobler}{Tobler}{1970}]{tobler1970computer}
Tobler, W.~R. (1970).
\newblock A computer movie simulating urban growth in the detroit region.
\newblock {\em Economic Geography\/}~{\em 46}, 234--240.

\bibitem[\protect\citeauthoryear{Wang and Ghosh}{Wang and
  Ghosh}{2023}]{wang2023nonparametric}
Wang, Y. and S.~K. Ghosh (2023).
\newblock Nonparametric estimation of isotropic covariance function.
\newblock {\em Journal of Nonparametric Statistics\/}~{\em 35\/}(1), 198--237.

\end{thebibliography}


\appendix

\section{Convergence analysis}
\label{app:conv}

Figure \ref{fig:traceplots} shows the traceplots for the two independently simulated chains, considering the post-burn-in period (5000 to 8000 iterations) for the intercept ($\beta_0$), the regression coefficients ($\beta_1$ and $\beta_2$), the variance parameter $\sigma^2$, and the dispersion parameter of the negative binomial ($\psi$).
\begin{figure}[H]
\begin{center}
	\includegraphics[height=8cm]{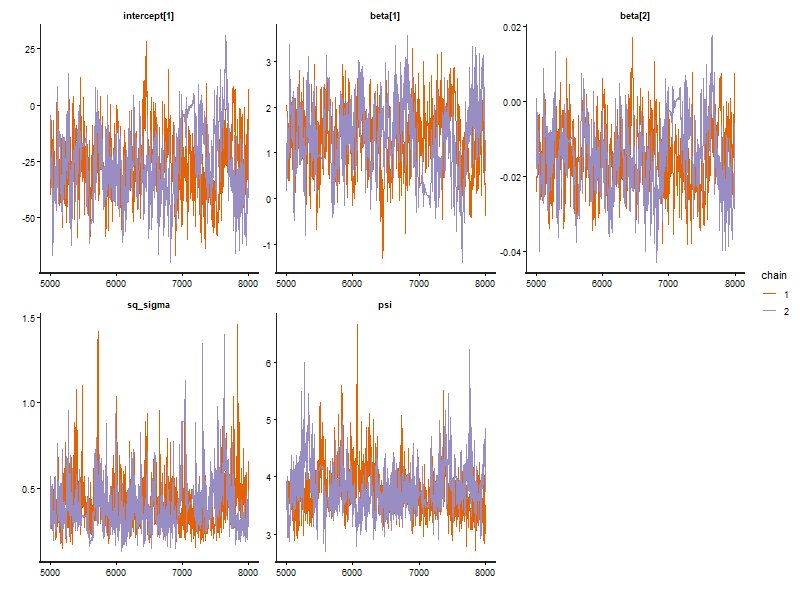}
	\caption{Traceplots of the MCMC chains for the structural parameters of the Negative GBP-Binomial model after the burn-in period.}\label{fig:traceplots}
\end{center}
\end{figure}
Visual analysis of Figure \ref{fig:traceplots} indicates that both chains exhibit good mixing and agreement across iterations. The trajectories fluctuate in a stable and stationary manner around a constant mean, showing neither upward nor downward trends nor prolonged periods of stagnation in specific regions of the parameter space. These visual diagnostics confirm the convergence of the MCMC chains, thereby validating subsequent inferences.

Table ~\ref{tab:conv_analysis} shows the presents the values of $\hat{R}$ and the effective sample size ($N_{eff}$).

\begin{table}[H]
\centering
\caption{Convergence criteria.}
\label{tab:conv_analysis}
\small
\begin{tabular}{lrrrrrr}
\toprule
\textbf{Parameter} & \textbf{$N_{eff}$} & \textbf{$\hat{R}$} \\
\midrule
$\beta_0$ &  92.280 & 1.006 \\
$\beta_1$ &  93.340 & 1.005 \\
$\beta_2$ &  96.040 & 1.004 \\
$\sigma^2$  & 173.921 & 1.006 \\
$\psi$ &  69.249 & 1.033 \\
\bottomrule
\end{tabular}
\end{table}

Additionally, it is observed that the sampling algorithm demonstrated good convergence, as the $\hat{R}$ diagnostic remained close to 1.00 for all parameters, accompanied by a satisfactory effective sample size ($N_{eff}$).



\end{document}